\pdfoutput=1
\documentclass[11pt]{article}
\usepackage{newtxtext}
\usepackage{amsthm}
\usepackage{newtxmath}
\usepackage[margin=1in]{geometry}
\usepackage{graphicx}
\usepackage{amsmath}
\usepackage{blkarray,booktabs,bigstrut}
\usepackage{multirow}
\usepackage{comment}
\usepackage{array}

\usepackage{enumerate}
\usepackage{chemarrow}
\usepackage{xcolor}

\usepackage{bbm}
\usepackage{subcaption}
\usepackage{makecell}
\usepackage[lined,ruled,linesnumbered,noend]{algorithm2e}

\usepackage{bm}
\usepackage{mathtools}
\usepackage{multicol,lipsum}
\usepackage[hidelinks]{hyperref}
\usepackage[nameinlink,capitalize]{cleveref} 

\def\thisismainpaper{0}   

\def\mbbE{\mathbb{E}}
\def\mbbN{\mathbb{N}}
\def\mbbZ{\mathbb{Z}}

\DeclareMathOperator{\argmin}{\arg\min}
\DeclareMathOperator{\argmax}{\arg\max}

\def\ddefloop#1{\ifx\ddefloop#1\else\ddef{#1}\expandafter\ddefloop\fi}
\def\ddef#1{\expandafter\def\csname bb#1\endcsname{\ensuremath{\mathbb{#1}}}}
\ddefloop ABCDEFGHIJKLMNOPQRSTUVWXYZ\ddefloop
\def\ddef#1{\expandafter\def\csname c#1\endcsname{\ensuremath{\mathcal{#1}}}}
\ddefloop ABCDEFGHIJKLMNOPQRSTUVWXYZ\ddefloop
\def\ddef#1{\expandafter\def\csname v#1\endcsname{\ensuremath{\boldsymbol{#1}}}}
\ddefloop ABCDEFGHIJKLMNOPQRSTUVWXYZabcdefghijklmnopqrstuvwxyz\ddefloop
\def\ddef#1{\expandafter\def\csname u#1\endcsname{\ensuremath{\underline{#1}}}}
\ddefloop ABCDEFGHIJKLMNOPQRSTUVWXYZabcdefghijklmnopqrstuvwxyz\ddefloop
\def\ddef#1{\expandafter\def\csname v#1\endcsname{\ensuremath{\boldsymbol{\csname #1\endcsname}}}}
\ddefloop {alpha}{beta}{gamma}{delta}{epsilon}{varepsilon}{zeta}{eta}{theta}{vartheta}{iota}{kappa}{lambda}{mu}{nu}{xi}{pi}{varpi}{rho}{varrho}{sigma}{varsigma}{tau}{upsilon}{phi}{varphi}{chi}{psi}{omega}{Gamma}{Delta}{Theta}{Lambda}{Xi}{Pi}{Sigma}{varSigma}{Upsilon}{Phi}{Psi}{Omega}{ell}\ddefloop

\makeatletter
\newcommand{\nosemic}{\renewcommand{\@endalgocfline}{\relax}}
\newcommand{\dosemic}{\renewcommand{\@endalgocfline}{\algocf@endline}}

\let\oldnl\nl
\newcommand{\nonl}{\renewcommand{\nl}{\let\nl\oldnl}}
\makeatother

\theoremstyle{definition}
\newtheorem{theorem}{Theorem}[section]
\newtheorem{lemma}[theorem]{Lemma}

\newcommand{\E}{{\rm I\kern-.3em E}}

\def\cmax{c_{\max}}

\allowdisplaybreaks
\title{Serving Chain-structured Jobs with Large Memory Footprints with Application to Large Foundation Model Serving}

\author{
Tingyang Sun\\
Pennsylvania State University\\
\texttt{tfs5679@psu.edu}
\and
Ting He\\
Pennsylvania State University\\
\texttt{tinghe@psu.edu}
\and
I-Hong Hou\\
Texas A\&M University\\
\texttt{ihou@tamu.edu}
}

\date{}

\begin{document}
\maketitle
\begin{abstract}
As a current trend in Artificial Intelligence (AI), large foundation models are increasingly employed as the core of AI services. However, even after training, serving such models at scale remains a challenging task due to their heavy resource footprints, particularly in terms of GPU memory. While recent works revealed unique characteristics of systems serving foundation models that distinguish them from traditional distributed computing systems, there is still a lack of fundamental understanding of the underlying system management problems. This work aims at addressing this gap by extracting a novel problem of ``server chain composition'' via block placement and cache allocation for serving chain-structured jobs with large memory footprints, which models a fundamental problem in serving large foundation models through pipeline parallelism. After showing the NP-hardness of the optimal solution, the focus is turned to developing scalable algorithms with guaranteed performance under state-of-the-art load balancing. Application of the proposed solution to a distributed large language model (LLM) serving system shows significant reduction of response times compared to state-of-the-art solutions.
\end{abstract}

\section{Introduction}\label{sec:Introduction}

Widely regarded as one of the most promising approaches towards realizing Artificial General Intelligence (AGI), large foundation models such as large language models (LLMs), large vision models, and related multimodal models are gaining popularity in both research communities and industries. The widespread deployment of such models introduces a system-level question of how to efficiently host these models and use them to answer queries (i.e., inference requests) at scale, given limited GPU resources. Due to the large sizes of foundation models and the high cost of high-end GPUs, \emph{model parallelism}, where a model is split across multiple GPUs/GPU servers, becomes a popular model serving approach. In particular, for physically distributed servers, the dominant approach is to split the model at the boundary of layers, a.k.a. \emph{pipeline parallelism}, which essentially places each model instance onto a chain of servers that are invoked sequentially during inference. 

Serving large foundation models is fundamentally different from serving traditional workloads. Most foundation models are based on the \emph{transformer architecture}, including all the LLMs (e.g., GPT series, LLaMA), vision foundation models such as Vision Transformer (ViT), and the backbone of most multimodal models (e.g., GPT-4V, Gemini). A large transformer has a massive memory footprint for storing both the model parameters (i.e., tensor weights) and the intermediate values generated during inference (i.e., KV values of the past tokens). This memory-intensive nature makes GPU memory the bottleneck resource in serving such models, which leads to fundamentally different resource allocation problems from serving traditional compute-bound workloads. 

In this work, we tackle a fundamental resource allocation problem underlying the emerging application of \emph{pipeline-parallel processing of inference requests by transformer-based large models}. Abstracting a given model as a chain of blocks (e.g., each representing a transformer layer) and the inference requests as jobs, we focus on the \emph{composition of server chains} by strategically placing the blocks onto servers and allocating cache space at each server for storing the intermediate values during job processing, to optimize the system performance in combination with a state-of-the-art load balancing policy. 
We further provide explicit performance bounds to not only characterize the performance of the resulting system but also guide the tuning of an important design parameter therein. 

\subsection{Related Work}\label{subsec:Related Work}

\textbf{Model-parallel model serving:}
One line of related works is on distributed serving of large models, a.k.a. model-parallel model serving. This includes \emph{pipeline parallelism} that splits the model at the granularity of layers~\cite{Huang19NeurIPS,Narayanan19SOSP,Yang21MLsys}, and \emph{tensor parallelism} that splits the model at the granularity of neurons~\cite{Krizhevsky17CommACM,Ben-Nun19ACMCompSurv,Tang20arXiv}. Tensor parallelism incurs a higher communication overhead across the splits and is thus typically limited within a single multi-GPU server~\cite{vLLM_distributed}. In contrast, pipeline parallelism only requires communications between adjacent splits and is thus more suitable for distributed settings. 
There are a few systems that support model-parallel model serving with focus on LLMs, such as vLLM with Ray Serve~\cite{vLLM_Ray}, Nvidia Dynamo~\cite{Nvidia_Dynamo}, Amazon EKS~\cite{Amazon_EKS}, and PETALS~\cite{Borzunov23NeurIPS}. Except for PETALS, most of these systems are designed for highly-connected environments (e.g., datacenters), while PETALS can work on weakly-connected devices. 
Most of the existing efforts focused on system implementation and ignored the fundamental problems in managing such systems, which is the gap addressed by this work. Such problems only started to receive attention recently, where \cite{Sun25Performance} considered a problem formulated after PETALS~\cite{Borzunov23NeurIPS} called ``block placement and request scheduling''. However, the solution therein contained many implementation-specific artifacts and was heuristic in nature. In this work, we will deepen this investigation by \emph{abstracting the fundamental problems in managing large model serving systems from a job processing perspective} and \emph{developing scalable solutions as well as performance analysis}. 
 
\looseness=0

\textbf{Service function chaining:}
Another line of related works is on \emph{service function chaining (SFC)}~\cite{Medhat17CommMag}. 
Under pipeline parallelism, each inference request is served by a chain of servers, which makes the resource allocation problem resemble the problem of \emph{service function (SF)} placement and flow routing in the context of SFC~\cite{Medhat17CommMag}. 
SFC placement and routing has been studied in the networking community, with domain-specific objectives (e.g., minimizing network operation cost or maximizing flow rate)~\cite{Chen19arXiv} and mostly heuristic solutions~\cite{Jang17JSAC,Chen19arXiv} (except for a few special cases such as \cite{Zhang17ICDCS,Guo18INFOCOM,Shang19ICPP}). Despite the conceptual similarity, \emph{SFC is fundamentally different from pipeline-parallel model serving}: (i) different SFs have heterogeneous resource requirements~\cite{Jang17JSAC,Zhang17ICDCS,Jalalitabar19NFV-SDN} and some SFs may change the flow rate~\cite{Ma17INFOCOM} or branch an input flow into multiple output flows~\cite{Jalalitabar19NFV-SDN}, while transformer layers of the same transformer model are typically identically-structured modules with the same resource requirements and input/output data size; (ii) more importantly, the bottleneck resource in large model serving is typically GPU memory, which differs from the bottleneck resource in SFC --- compute cycles --- in how it is consumed by the workloads\footnote{The part of GPU memory used to store model parameters is shared by all the processed requests while the part used to store KV values for each request is dedicated to this request, differing from the additive consumption of compute cycles~\cite{Zhang17ICDCS,Guo18INFOCOM}.}.

\textbf{Load balancing:} Viewing each inference request as a job, our problem can be viewed as a job processing problem where the processing is distributed to a chain of servers. In a multi-server context, most existing works focused on load balancing with given servers, including Join-the-Shortest-Queue (JSQ)~\cite{Winston77JAP}, Join-the-Idle-Queue (JIQ)~\cite{Lu11PE}, and their low-overhead alternatives (e.g., JSQ-$d$~\cite{Mitzenmacher96FOCS} and JIQ-$d$~\cite{Wang18TON}) for homogeneous servers, and later adaptations like Speed-Aware JSQ (SA-JSQ)~\cite{Bhambay22PE} (a.k.a. Join-the-Fastest-of-the-Shortest-Queue (JFSQ)~\cite{Weng20MACS}) for heterogeneous servers. However, in our case, a major challenge is how to ``compose the servers'' by distributing model parameters and allocating KV caches, as serving an inference request not only requires a chain of servers hosting all the model parameters but also enough cache space at each server for the intermediate values. 
A related problem of job scheduling in LLM serving was discussed in \cite{Mitzenmacher25SS}, which considered scheduling inference requests to a single LLM instance; in our case, the deployed blocks may form many (partially overlapping) model instances and the solutions in \cite{Mitzenmacher25SS} can be used for scheduling within each instance 
after dispatching. 
A seemingly related problem is ``job shop scheduling''~\cite{Gromicho12COR}, where each job has a chain of operations to be performed sequentially by a set of machines and each machine can only perform one operation at a time. However, our problem differs not only in that our machine-operation association is controllable (through block placement),  
but also in that in our case the operations of the same job need to be scheduled simultaneously instead of sequentially\footnote{This is because although microscopically the forward pass in generating each token is processed sequentially by the layers, macroscopically the inference session for serving each request needs to be run simultaneously at all the servers in the selected server chain to perform sequence generation in an auto-regressive manner.}. 
To our knowledge, the problem of \emph{processing chain-structured memory-bound jobs by composable server chains} has not been formally studied before. 

\subsection{Summary of Contributions}
We tackle a novel problem of processing chain-structured memory-bound jobs by composable server chains, inspired by pipeline-parallel processing of inference requests for large transformer models, with the following contributions:
\begin{enumerate}
    \item We decompose the resource allocation problem in such a system into three coupled subproblems, targeting at optimizing (i) the \emph{placement of service blocks} across servers to form server chains, (ii) the \emph{allocation of cache space} at each server to determine the (concurrency) capacity of each server chain, and (iii) the \emph{dispatching of jobs} to the server chains. The first two subproblems are solved offline to compose ``job servers'' that can each serve a job independently, and the third is solved online to load balance among the ``job servers''. 

    \item The main difficulty resides in the first two subproblems, both of which are NP-hard and contend for the limited server memory. Under a state-of-the-art load balancing policy of ``Join-the-Fastest-Free-Server'', we develop a set of efficient algorithms for the first two subproblems that achieve optimality in important special cases, and provide explicit analysis in terms of upper/lower bounds on the steady-state mean response time. 

    \item We test our proposed solution against state-of-the-art resource allocation algorithms for distributed LLM serving via both model-driven simulations and experiments based on a real LLM serving system. Besides validating our theoretical analysis, these tests also show that the proposed solution can provide significant performance improvement (with $63$--$77\%$ reduction in mean response time) under real demands deviating from our original assumptions, which indicates the robustness of our solution. 
     
\end{enumerate}

\textbf{Roadmap.} Section~\ref{sec:Background and Formulation} presents our problem formulation, Section~\ref{sec:Optimized Job Serving} describes the proposed solution and its analysis, Section~\ref{sec:Performance Evaluation} tests the proposed solution against benchmarks in serving LLM inference requests, and Section~\ref{sec:Conclusion} concludes the paper with a discussion. 
\textbf{All the proofs are provided in Appendix~\ref{appendix:Proofs}.}

\section{Problem Formulation}\label{sec:Background and Formulation}

\subsection{System Model}\label{subsec:System Model}

\subsubsection{Abstract Model}
Consider a system consisting of a set $\mathcal{J}$ of $J:=|\mathcal{J}|$ servers. 
Each server $j\in \mathcal{J}$ has a \emph{memory size} of $M_j$ and incurs a \emph{mean communication time} of $\tau^c_j$ to participate in a job's processing and a \emph{mean computation time} of $\tau^p_j$ to process each service block for the job. 

The system hosts a service consisting of $L$ blocks, each with \emph{block size} $s_m$, to serve dynamically arriving jobs. When serving a job, a server needs to cache the intermediate results from each block it processes for the job, with a \emph{cache size} of $s_c$ per block per job. To be successfully served, a job needs to (i) be assigned to a chain of servers that collectively host all the $L$ blocks \emph{in order} and (ii) be allocated enough cache space at each server. Due to the communication overhead, we consider block placements that place a continuous range of blocks at each server, represented by $\{a_j,\ldots,a_j+m_j-1\}$ for server $j\in\mathcal{J}$, where $a_j$ denotes the \emph{index of the first block} and $m_j$ the \emph{number of blocks}, and thus the entire block placement can be represented by $(\bm{a}:=(a_j)_{j\in \mathcal{J}},\: \bm{m}:=(m_j)_{j\in \mathcal{J}})$. Requirement (i) means that servers $i, j\in \mathcal{J}$ can be traversed consecutively by a server chain if and only if $a_j\leq a_i+m_i\leq a_j+m_j-1$, and requirement (ii) means that for each job that is processed by a server chain traversing $(i,j)$ (i.e., traversing servers $i$ and $j$ consecutively), server $j$ needs to allocate a cache space of $s_c(a_j+m_j-a_i-m_i)$, assuming that each block is processed at the first server hosting it\footnote{This is consistent with PETALS~\cite{Borzunov23NeurIPS}. Other tie-breaking rules can be modeled by modifying the number of processed blocks accordingly.}. 

To handle the boundary cases at the beginning/end of each chain, we introduce two \emph{dummy servers} $j_0$ and $j_{J+1}$, with $a_{j_0}:=0$, $a_{j_{J+1}} := L+1$, and $m_{j_0}=m_{j_{J+1}} := 1$, i.e., $j_0$ hosts a dummy block $0$ and $j_{J+1}$ hosts a dummy block $L+1$. We denote the extended server set by $\mathcal{J}_+:= \mathcal{J}\cup \{j_0, j_{J+1}\}$. 
For the ease of presentation, we denote by $m_{ij}:= a_j+m_j-a_i-m_i$ the \emph{number of blocks processed at server $j$} after the processing at server $i$, $\mathcal{E}_{\bm{a},\bm{m}}:=\{(i,j)\in \mathcal{J}_+^2:\: a_j\leq a_i+m_i\leq a_j+m_j-1\}$ the \emph{set of possible neighbors} on server chains, and $\mathcal{K}_{\bm{a},\bm{m}}$ the set of feasible server chains, all under a given block placement $(\bm{a},\bm{m})$. To ensure the traversal of all the $L$ blocks in order, each chain in $\mathcal{K}_{\bm{a},\bm{m}}$ must start from $j_0$, end at $j_{J+1}$, and only traverse links in $\mathcal{E}_{\bm{a},\bm{m}}$. 
\if \thisismainpaper1
The main notation is summarized in Appendix~A.1 of \cite{Sun26arXiv}.
\else
The main notation is summarized in Table~\ref{tab:notations} of Appendix~\ref{appendix:Notations}. 
\fi

The above system model implies that under a block placement $(\bm{a}, \bm{m})$, a server $j$ that simultaneously processes $c_{ij}$ jobs for the server chains traversing $(i,j)$ has a \emph{total memory consumption} of 
\begin{align}\label{eq:memory consumption model}
s_m m_j + s_c \sum_{i:\: (i,j)\in \mathcal{E}_{\bm{a},\bm{m}}} c_{ij}m_{ij},
\end{align}
and a job assigned to a server chain $k$ has a \emph{mean service time} of 
\begin{align}\label{eq:mean service time model}
T_k:= \sum_{(i,j)\in k}\left(\tau^c_j + \tau^p_j m_{ij} \right),
\end{align}
where $(i,j)\in k$ means that chain $k$ traverses servers $i$ and $j$ consecutively. 
While processing a job also requires other types of resources (e.g., compute cycles and network bandwidth), we assume the workload to be \emph{memory-bound}, i.e., memory is the bottleneck resource that will be exhausted before other resources.  
To ensure that dummy servers do not contribute to the service time or impose resource constraints, we assume $\tau^c_{j_{J+1}} = \tau^p_{j_{J+1}}:= 0$ and $M_{j_{J+1}}:=\infty$. 

\subsubsection{Connection to Large Model Serving} 

The above system model is inspired by systems used to serve inference requests for transformer-based large models through pipeline parallelism, where the service represents a large transformer with decoder-only architecture that is used by many state-of-the-art large foundation models (e.g., GPTs).  

Specifically, each ``job'' represents an \emph{inference request} that aims at generating a sequence of output tokens based on a sequence of input tokens, where most of the processing is carried out in a set of identically structured \emph{transformer layers}. Each ``service block'' represents such a transformer layer, and each ``cache'' represents a \emph{KV cache} used to store the context information in the form of KV values of the past tokens (including input tokens) processed for a given request by a given transformer layer. Serving a job includes both a \emph{prefill phase}, which populates the caches with KV values from processing the input tokens, and a \emph{decode phase}, which generates the output tokens one by one by passing the latest token through the transformer layers auto-regressively. 
We assume static cache allocation, which pre-allocates a fixed-size cache for each transformer layer at the beginning of each job processing according to a predetermined maximum sequence length. We assume that once a job starts its processing at a chain of servers, it must run there until completion (i.e., no migration or preemption), which holds in most practical systems due to the high overhead in transferring/swapping KV caches. 
The mean communication time $\tau^c_j$ represents the mean total time for communicating the input/output to server $j$ in generating all the tokens  for a request\footnote{This model fits the communication method in PETALS~\cite{Borzunov23NeurIPS} with a frontend server as the proxy of clients, where the frontend server connects the servers in a chain by sending the input to a server and then forwarding its output to the next server. Our solution can be adapted to the case where consecutively traversed servers directly communicate as in \cite{vLLM_Ray}, by making the communication time dependent on the preceding server, i.e., $\tau^c_{ij}$, and taking $\tau^c_j:= \max_i \tau^c_{ij}$.}, and the mean per-block computation time $\tau^p_j$ represents the mean total time in processing all the tokens for a request by a transformer layer. 

While this system model omits some possible complications\footnote{For example, the KV cache size can vary across requests and even grow during the processing of the same request if dynamic allocation is used.}, it captures the main difference between serving inference requests for large transformers and traditional workloads: \emph{the shift of bottleneck resource from compute/network to memory}, induced by the large memory footprints of transformers, the relatively small runtime bandwidth demands, 
and the massive parallel processing capabilities of GPUs. 
This shift leads to a novel challenge in resource allocation: as (GPU) memory is a critical resource for both placing blocks and processing inference requests (due to the need of KV caches), there is a \emph{resource contention between the model placement stage and the request serving stage}, as the residual memory after model placement determines how many requests each server can process in parallel. We will address this contention explicitly.

\subsection{Optimization Problems}\label{subsec:Optimization Problems}

Our overall goal is to optimize the system's performance by strategically composing server chains and using the composed chains to process incoming jobs. In this work, we assume a \emph{centralized orchestrator} that controls these decisions and serves as the ingress/egress point of jobs; distributed variations are left to future work. 
We focus on services with a large memory footprint (e.g., large foundation models), where loading service blocks to servers takes substantial time. This motivates us to consider a two-time-scale approach: \emph{server chain composition} via block placement and cache allocation at a large time scale, and \emph{load balancing} for assigning incoming jobs to the composed server chains at a small time scale.

\emph{Offline Server Chain Composition:}
As mentioned in Section~\ref{subsec:Related Work}, one of the main differences between our problem and classical job serving problems is that the ``job servers'' in our case, each representing a chain of servers with resources for serving a job independently, are composable. Their composition depends on not only how the blocks are placed onto the physical servers, but also how each server allocates its residual memory among the server chains traversing it as caches for ongoing jobs. 

Specifically, since under a given block placement $(\bm{a},\bm{m})$, each ongoing job traversing $(i,j)\in \mathcal{E}_{\bm{a},\bm{m}}$ (i.e., processed by servers $i$ and $j$ consecutively) consumes $m_{ij}$ \emph{cache slots} at server $j$ (each of size $s_c$), we can represent the cache allocation by a vector $\bm{c}:=(c_k)_{k\in \mathcal{K}_{\bm{a},\bm{m}}}$ that specifies the maximum number of jobs processed concurrently on all the feasible server chains. We refer to $c_k$
as the \emph{capacity of chain $k$}. Then $m_{ij} c_k$ is the number of cache slots server $j$ allocates to a chain $k$ that traverses it from server $i$. According to the memory consumption model \eqref{eq:memory consumption model}, the cache allocation $\bm{c}$ must satisfy
\begin{align}\label{eq:tildeM_j}
\sum_{i:\: (i,j)\in \mathcal{E}_{\bm{a},\bm{m}}} m_{ij} \sum_{k:\: (i,j)\in k}c_k \leq \left\lfloor {M_j-s_m m_j\over s_c} \right\rfloor =: \widetilde{M}_j,~~~\forall j\in \mathcal{J},
\end{align}
where $\widetilde{M}_j$ denotes the number of cache slots at server $j$  after hosting $m_j$ blocks. 

From a job processing perspective, each feasible server chain $k$ with capacity $c_k$ can process up to $k$ jobs in parallel (with possibly different processing times), we model such a server chain as $c_k$ ``virtual servers'' that can each serve one job at a time in the queueing-theoretic sense. We refer to such virtual servers as \emph{job servers} to differentiate from the physical servers. Let $\mu_k:=1/T_k$ denote the \emph{service rate of chain $k$}, where $T_k$ is the mean service time defined in \eqref{eq:mean service time model}. Then the block placement $(\bm{a},\bm{m})$ determines the set of feasible server chains $\mathcal{K}_{\bm{a},\bm{m}}$ and their service rates $(\mu_k)_{k\in \mathcal{K}_{\bm{a},\bm{m}}}$, while the cache allocation determines the capacity of each chain $(c_k)_{k\in \mathcal{K}_{\bm{a},\bm{m}}}$. Together, they compose a set of job servers $(\mu_k,c_k)_{k\in \mathcal{K}_{\bm{a},\bm{m}}}$, where there are $c_k$ job servers of service rate $\mu_k$.

\emph{Online Load Balancing:}
Given the job servers composed by block placement and cache allocation, the remaining problem reduces to a classical online load balancing problem where the orchestrator dispatches incoming jobs to the composed job servers, which are generally heterogeneous in terms of service rates. 
This suggests that we can borrow from the existing solutions on load balancing for heterogeneous servers such as \cite{Weng20MACS,Bhambay22PE}. 
A subtle difference in our case is that the job servers are logical servers instead of physical servers, and thus any arrived job will be kept in a central queue by the orchestrator before processing. Although any work-conserving load balancing policy can stabilize the system for arrival rates within the \emph{total service rate}, defined as
\begin{align}\label{eq:total service rate}
\nu := \sum_{k\in \mathcal{K}_{\bm{a},\bm{m}}}\mu_k c_k,
\end{align}
different policies can lead to different response times. 

\emph{Design Objective:}
Given a request arrival rate $\lambda$, our goal is to jointly design the block placement, the cache allocation, and the load balancing policy so as to \emph{minimize the mean response time} (including waiting and processing), while ensuring queue stability. 

\emph{Remark:} In this work, we assume the standard First-Come-First-Serve (FCFS) scheduling to focus on the above optimizations. Our solution can be combined with optimizations of job scheduling as in \cite{Mitzenmacher25SS} for further improvement, which is left to future work. \looseness=0

\section{Optimized Job Serving}\label{sec:Optimized Job Serving}

We now address the optimization problems identified in Section~\ref{subsec:Optimization Problems}, with focus on developing \emph{scalable algorithms}.

\subsection{Server Chain Composition}\label{subsec:Server Chain Composition}

\subsubsection{Optimization Formulation}\label{subsubsec:BPCA Optimization}
As explained in Section~\ref{subsec:Optimization Problems}, the block placement $(\bm{a},\bm{m})$ and the cache allocation $\bm{c}$ jointly determine the feasible server chains and their capacities, which provides the job servers for subsequent processing. Since the ultimate performance depends on how jobs are assigned to the composed job servers, ideally we should design $(\bm{a},\bm{m})$ and $\bm{c}$ with the load balancing policy in mind. Let $\mathcal{A}$ denote the adopted load balancing policy with a mean response time of $\overline{T}_{\mathcal{A}}((\mu_k,c_k)_{k\in \mathcal{K}_{\bm{a},\bm{m}}})$. Then our original goal is to $\min_{\bm{a},\bm{m},\bm{c}} \overline{T}_{\mathcal{A}}((\mu_k,c_k)_{k\in \mathcal{K}_{\bm{a},\bm{m}}})$. However, directly solving this optimization is intractable due to the implicit and complicated dependency of the objective on the decision variables. We thus simplify the problem as follows. 

Let $\overline{\rho}\in (0,1)$ be a given {safety margin} for stability that represents a \emph{target maximum system load}. For a sufficiently small $\overline{\rho}$, the response time will be dominated by the service time, which allows us to focus on minimizing the mean service time, given by 
\begin{align}
\overline{T}_{\mathcal{A}}((\mu_k,c_k)_{k\in \mathcal{K}_{\bm{a},\bm{m}}}) \approx \sum_{k\in \mathcal{K}_{\bm{a},\bm{m}}} {1\over \mu_k}\cdot {\lambda_{\mathcal{A},k}\over \lambda}, \label{eq:mean service time}
\end{align}
where $\lambda_{\mathcal{A},k}$ is the rate for policy $\mathcal{A}$ to assign requests to chain $k$. Since any policy $\mathcal{A}$ preserving stability must satisfy $\lambda_{\mathcal{A},k}\leq c_k \mu_k$, we can upper-bound \eqref{eq:mean service time} by
\begin{align}
\overline{T}_{\mathcal{A}}((\mu_k,c_k)_{k\in \mathcal{K}_{\bm{a},\bm{m}}}) \leq \sum_{k\in \mathcal{K}_{\bm{a},\bm{m}}} {1\over \mu_k}\cdot {c_k\mu_k \over \lambda} = {1\over \lambda} \sum_{k\in \mathcal{K}_{\bm{a},\bm{m}}} c_k. \label{eq:mean service time bound}
\end{align}
%
This bound allows us to formulate the block placement and cache allocation problem as an explicit optimization as follows\footnote{We use $[n]$ to denote the set $\{1,\ldots,n\}$ for a positive integer $n$.}
\begin{subequations}\label{eq:BPCA - conceptual}
\begin{align}
\min_{\bm{a},\bm{m},\bm{c}} \quad & \sum_{k\in \mathcal{K}_{\bm{a},\bm{m}}} c_k \label{BPCA:obj}\\
\mbox{s.t.}\quad 
& \sum_{k\in \mathcal{K}_{\bm{a},\bm{m}}} {c_k\over \sum_{(i,j)\in k}(\tau^c_j+\tau^p_j m_{ij})} \geq {\lambda\over \overline{\rho}}, \label{BPCA:service rate}\\
& a_j+m_j-1\leq L,~~\forall j\in \mathcal{J},\label{BPCA:placement}\\
& \sum_{i: (i,j)\in \mathcal{E}_{\bm{a},\bm{m}}}m_{ij}\sum_{k: (i,j)\in k} c_k \leq \widetilde{M}_j,~~\forall j\in \mathcal{J}, \label{BPCA:memory}\\
& a_j, m_j\in [L],\: c_k\in \mbbN,
\end{align}
\end{subequations}
where \eqref{BPCA:obj} and \eqref{BPCA:service rate} imply the objective of minimizing an upper bound on the mean response time under stability constraints (with safety margin), \eqref{BPCA:placement} ensures validity of the placement variables, and \eqref{BPCA:memory} enforces the memory constraints. 
Recall that $m_{ij}$ is the number of blocks processed at server $j$ for a chain traversing $(i,j)$, $\widetilde{M}_j$ is the number of cache slots at server $j$, and $\mathcal{K}_{\bm{a},\bm{m}}$/$\mathcal{E}_{\bm{a},\bm{m}}$ is the set of server chains/pairs feasible under block placement $(\bm{a},\bm{m})$.

\subsubsection{NP-hardness}
The optimization \eqref{eq:BPCA - conceptual} is a complex integer programming problem with nonlinear constraints. Moreover, the number of variables in $\bm{c}$ scales with the number of feasible chains $|\mathcal{K}_{\bm{a},\bm{m}}|$, which is generally exponential in the size of the system. 
Rigorously, we can prove that even though the block placement is given, the remaining cache allocation problem is still NP-hard. 

\begin{theorem}\label{thm:NP-hardness of BPCA}
Computing the optimal cache allocation as in \eqref{eq:BPCA - conceptual} is NP-hard even if the server chains are given. 
\end{theorem}

\emph{Remark:} Theorem~\ref{thm:NP-hardness of BPCA} implies that even if the block placement is given, optimizing server chain composition remains a hard problem.

\subsubsection{Block Placement}\label{subsubsec:Block Placement Algorithm}

Intuitively, block placement depends on the cache allocation strategy as the latter defines how the placed blocks will be used to form job servers. Since the cache allocation subproblem is already NP-hard by Theorem~\ref{thm:NP-hardness of BPCA}, we need a simpler (but possibly suboptimal) cache allocation strategy to facilitate the optimization of  block placement. In particular, the proof of Theorem~\ref{thm:NP-hardness of BPCA} implies that the cause of the hardness of cache allocation is the sharing of servers across chains. This inspires us to restrict our solution space to \emph{disjoint server chains} that can each host the entire set of blocks with sufficient capacity.  

\emph{Simplification:}
Given a \emph{required capacity} $c\in \mbbZ^+$ denoting the minimum number of jobs each server needs to serve concurrently, the maximum number of blocks we can place at server $j$ is
\begin{align}\label{eq:m_j(c)}
m_j(c) := \min\left(\left\lfloor{M_j\over s_m+s_c c}\right\rfloor,\: L \right),
\end{align}
obtained by reserving $c$ cache slots for each placed block. 
Accordingly, the mean service time each job spends at this server is upper-bounded by
\begin{align}
t_j(c) := \tau^c_j + \tau^p_j m_j(c),
\end{align}
achieved when all the $m_j(c)$ blocks are processed. 
For now we treat $c$ as a given parameter, which will be optimized later (see \eqref{eq:optimal c}). 

Forming disjoint server chains is equivalent to selecting disjoint subsets of servers such that each subset has enough memory to host all the blocks. Under cache allocation according to a required capacity $c$, the resulting total service rate can be bounded as follows. 

\begin{lemma}\label{lem:service rate bound}
Given a collection of disjoint subsets of servers $(\mathcal{J}_k)_{k\in\mathcal{K}}$ such that $\sum_{j\in \mathcal{J}_k}m_j(c)\geq L$ for all $k\in \mathcal{K}$, the block placement that places the entire set of $L$ blocks onto each subset $\mathcal{J}_k$ with $m_j(c)$ blocks at each server $j\in \mathcal{J}_k$ can achieve a total service rate of no less than $c\sum_{k\in \mathcal{K}}(\sum_{j\in \mathcal{J}_k}t_j(c))^{-1}$.  
\end{lemma}

By restricting the solution space to disjoint server chains and reserving cache space according to the required capacity, we can simplify the joint optimization problem in \eqref{eq:BPCA - conceptual} as follows:\looseness=0
\begin{subequations}\label{eq:BP}
\begin{align}
\min_{(\mathcal{J}_k)_{k\in \mathcal{K}}}\quad & |\mathcal{K}| \label{BP:obj}\\
\mbox{s.t.}\quad & \sum_{k\in \mathcal{K}} {1\over \sum_{j\in \mathcal{J}_k}t_j(c)} \geq {\lambda\over \overline{\rho} c}, \label{BP:service rate} \\
& \sum_{j\in \mathcal{J}_k}m_j(c) \geq L,~~~\forall k\in \mathcal{K}, \label{BP:feasibility} \\
& \mathcal{J}_k\cap \mathcal{J}_{k'}=\emptyset,~~~\forall k,k'\in \mathcal{K},\: k\neq k', \label{BP:disjoint} \\
& \mathcal{J}_k \subseteq \mathcal{J},~~~\forall k\in \mathcal{K}, \label{BP:subset}
\end{align}
\end{subequations}
which aims at grouping the servers into the smallest number of disjoint subsets (ensured by \eqref{BP:disjoint}--\eqref{BP:subset}), subject to constraint \eqref{BP:feasibility} that ensures each subset to form a feasible server chain of capacity $c$, and constraint \eqref{BP:service rate} that ensures a sufficiently large total service rate according to Lemma~\ref{lem:service rate bound}. 

\emph{Algorithm:}
The simplified problem in \eqref{eq:BP} is still NP-hard. 

\begin{lemma}\label{lem:NP-hardness of server grouping}
The optimization in \eqref{eq:BP} is NP-hard. 
\end{lemma}

\begin{algorithm}[tb]
\small
\SetKwInOut{Input}{input}\SetKwInOut{Output}{output}
\Input{Server set $\mathcal{J}$, required capacity $c$, \#blocks $m_j(c)$ and mean service time  $t_j(c)$ for each $j\in \mathcal{J}$, total \#blocks $L$, required (scaled) total service rate $\lambda/(\overline{\rho}c)$}
\Output{Block placement $(\bm{a}(c),\bm{m}(c))$}
$\widetilde{t}_j(c) \leftarrow t_j(c)/m_j(c)$, $\forall j\in \mathcal{J}$\;
$a\leftarrow 1$, $\nu\leftarrow 0$, $T\leftarrow 0$\;
\For{each server $j\in \mathcal{J}$ in increasing order of $\widetilde{t}_j(c)$ \label{GBP:3}}
{
$a_j(c) \leftarrow \min(a,\: L-m_j(c)+1)$\; \label{GBP:4}
$T\leftarrow T + t_j(c)$\;
$a\leftarrow \min(a+m_j(c)-1,\: L)+1$\;
\If{$a>L$}
{
$\nu\leftarrow \nu+1/T$\; \label{GBP:8}
\If{$\nu\geq \lambda/(\overline{\rho}c)$ \label{GBP:9}}
{break\;}
\Else
{$a\leftarrow 1$, $T\leftarrow 0$\;}
}
}
\caption{Greedy Block Placement with Cache Reservation (GBP-CR)}
\vspace{-.0em}
\label{Alg:GBP-CR}
\end{algorithm}
\normalsize

Nevertheless, this simplified problem formulation inspires an intuitive block placement algorithm. The basic intuition is that ``chaining fast servers together is better than mixing fast and slow servers''. For example, if $\exists$two subsets of servers $\mathcal{J}_1,\mathcal{J}_2$ not satisfying \eqref{BP:feasibility} yet with $\sum_{j\in \mathcal{J}_1}t_j(c) < \sum_{j\in \mathcal{J}_2}t_j(c)$, and two servers $j_1,j_2$ not in either subset that satisfy $m_{j_i}(c)\geq \max(L-\sum_{j\in \mathcal{J}_1}m_j(c),\: L-\sum_{j\in\mathcal{J}_2}m_j(c))$ ($i=1,2$) and $t_{j_1}(c)<t_{j_2}(c)$, then one can verify that 
\begin{align}
&{1\over \sum_{j\in \mathcal{J}_1}t_j(c) + t_{j_1}(c)} + {1\over \sum_{j\in \mathcal{J}_2}t_j(c) + t_{j_2}(c)} > \nonumber\\
&\hspace{6em} {1\over \sum_{j\in \mathcal{J}_1}t_j(c) + t_{j_2}(c)} + {1\over \sum_{j\in \mathcal{J}_2}t_j(c) + t_{j_1}(c)}. 
\end{align}
Thus, we can always achieve a higher total service rate by a given number of chains by adding a faster server to a faster chain and a slower server to a slower chain. 
Since different servers can host different numbers of blocks, we compare server speeds by the \emph{amortized mean service time} per block, given by 
\begin{align}\label{eq:amortized service time}
\widetilde{t}_j(c) := {t_j(c)\over m_j(c)}.
\end{align}
These ideas lead to a greedy block placement algorithm referred to as \emph{Greedy Block Placement with Cache Reservation (GBP-CR)}, shown in Alg.~\ref{Alg:GBP-CR}, which sequentially places blocks onto servers in descending speeds to form server chains (lines~\ref{GBP:3}--\ref{GBP:4}), until all the servers are processed or the scaled total service rate (denoted by $\nu$) satisfies the given requirement (line~\ref{GBP:9}). 

\emph{Complexity:} Alg.~\ref{Alg:GBP-CR} has a near-linear time complexity of $O(J\log{J})$, dominated by the sorting of servers in line~\ref{GBP:3}. 


\emph{Analysis:}
As an efficient solution, GBP-CR is expected to be generally suboptimal for \eqref{eq:BP} due to the NP-hardness of the problem. 
Nevertheless, GBP-CR is optimal in a special case of practical relevance. Specifically, consider the case where all the servers have an \emph{identical memory size} $M_j\equiv M$ ($\forall j\in \mathcal{J}$). This can be caused by servers having a standard memory configuration (e.g., 8 GB for consumer GPUs or 80 GB for datacenter GPUs) or owners limiting the amount of memory made available to the model serving system even if the servers may have more physical memory (e.g., when the GPUs are contributed by volunteers as in the case of PETALS~\cite{Borzunov23NeurIPS}). In this special case, GBP-CR is optimal in the following sense. 

\begin{theorem}\label{thm:optimality of GBP-CR}
In the case of homogeneous memory, i.e., $M_j\equiv M$ ($\forall j\in \mathcal{J}$), GBP-CR (Alg.~\ref{Alg:GBP-CR}) provides an optimal solution to \eqref{eq:BP}. 
\end{theorem}

\emph{Parameter tuning:}
The performance of GBP-CR (Alg.~\ref{Alg:GBP-CR}) crucially depends on its input parameter $c$, which controls a ``service time vs. waiting time''
tradeoff: a smaller $c$ allows more blocks to be placed at each server, which helps to shorten server chains and thus reduce service times; meanwhile, a larger $c$ leaves the servers with more residual memory after block placement, which allows them to run more jobs concurrently that helps to reduce waiting times. 

\if\thisismainpaper 0
\begin{figure}[!t]
   \centerline{\includegraphics[width=0.46\linewidth]{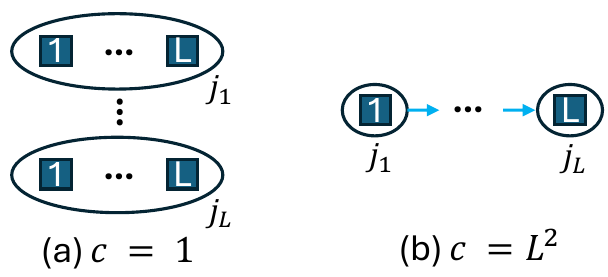}}
    \caption{Example for the impact of the capacity requirement $c$: for $\mathcal{J}=\{j_1,\ldots,j_L\}$, $M_j\equiv (L+1)s_m$, $s_m=L s_c$, $\tau^c_j\equiv \tau^c$, and $\tau^p_j\equiv \tau^p$, (a) shows the chains constructed by Alg.~\ref{Alg:GBP-CR} for $c=1$, and (b) shows the chains constructed by Alg.~\ref{Alg:GBP-CR} for $c=L^2$. 
    }
    \label{fig:throughput_delay_tradeoff}
\end{figure}

To illustrate this tradeoff, consider the example in Fig.~\ref{fig:throughput_delay_tradeoff}, where the squares denote blocks and the circles denote servers. Under $c=1$, we have $m_j(c)\equiv L$, which leads to $L$ single-server chains as shown in Fig.~\ref{fig:throughput_delay_tradeoff}a, each of capacity $1$. This configuration yields a mean service time of $T^{(1)} = \tau^c+L \tau^p$ and a total service rate of $\nu^{(1)}=L/(\tau^c+L \tau^p)$. Under $c=L^2$, we have $m_j(c)\equiv 1$, which leads to a single $L$-server chain as shown in Fig.~\ref{fig:throughput_delay_tradeoff}b with capacity $L^2$. This configuration yields a mean service time of $T^{(2)} = L\tau^c + L\tau^p$ and a total service rate of $\nu^{(2)} = L/(\tau^c+\tau^p)$. We see that $T^{(1)}< T^{(2)}$ but $\nu^{(2)}>\nu^{(1)}$, i.e., the second configuration is more suitable for high demands and the first configuration is more suitable for low demands. This indicates the existence of a demand-dependent optimal configuration of $c$. 
\fi

Let $\cmax := \lfloor (\max_j M_j - s_m)/s_c \rfloor$ denote the maximum number of concurrent jobs supported by any server. Assuming GBP-CR as the block placement algorithm, we can measure the performance under a given configuration of $c$ by the objective value of \eqref{eq:BP} achieved by GBP-CR for a required capacity of $c$, given by 
\begin{align}
K(c) := \min \{K:\: \sum_{l=1}^K {1\over \sum_{j\in k_l(c)} t_j(c)}\geq {\lambda\over \overline{\rho}c}\},
\end{align}
where $k_l(c)$ denotes the $l$-th server chain constructed by GBP-CR under capacity $c$. Then the optimal configuration is given by
\begin{align}\label{eq:optimal c}
c^* := \argmin_{c\in [\cmax]} c\cdot K(c).
\end{align}
By the argument in Section~\ref{subsubsec:BPCA Optimization}, such a configuration will minimize a surrogate of the mean response time if the jobs are processed by the disjoint server chains constructed by GBP-CR with a given capacity. 
The optimization \eqref{eq:optimal c} is easy to solve via brute-force search due to its small 1D solution space. 


\subsubsection{Cache Allocation}\label{subsubsec:Cache Allocation}

\if\thisismainpaper 1
Although during block placement, GBP-CR already reserves cache space at each server for processing $c$ jobs concurrently, the residual memory often allows more concurrency under proper cache allocation
as illustrated in Fig.~2 in \cite{Sun26arXiv}, which suggests the need to further optimize server chain composition after block placement. 
\else
\begin{figure}[!t]
   \centerline{\includegraphics[width=0.6\linewidth]{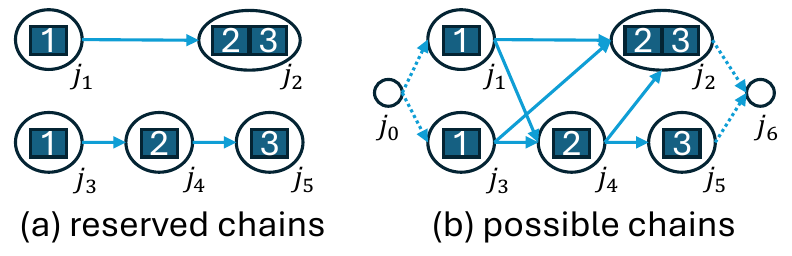}}
    \caption{Opportunity for further optimization after block placement with cache reservation: for $\mathcal{J}=\{j_1,\ldots,j_5\}$, $L=3$, $s_m=1$, $s_c=0.1$, $M_j=3$ if $j=j_2$ and $2$ otherwise, $\tau^c_j=2$ if $j=j_2$ and $1$ otherwise, and $\tau^p_{j_l}=l\epsilon$ for $0<\epsilon\ll 1$, (a) shows the chains constructed by Alg.~\ref{Alg:GBP-CR} for $c=1$, and (b) shows all possible chains under the same block placement.   
    }
    \label{fig:example_cache_allocation}
\end{figure}

Although during block placement, GBP-CR already reserves cache space at each server for processing $c$ jobs concurrently, the residual memory often allows more concurrency via proper cache allocation. 

\emph{Example:} Consider the problem instance in Fig.~\ref{fig:example_cache_allocation}, where the amortized mean service time for server $j_l$ ($l\in [5]$) under $c=1$ is $\widetilde{t}_{j_l}(c)=1+l\epsilon$. GBP-CR will form two disjoint server chains $k_1 = j_1\to j_2$ and $k_2=j_3 \to j_4 \to j_5$ as shown in Fig.~\ref{fig:example_cache_allocation}a, each with a capacity of $c=1$, which yields a total service rate of
\begin{align}
\nu = {1\over 3+5\epsilon} + {1\over 3+12\epsilon}.
\end{align}
Meanwhile, any directed path traversing all the 3 blocks is a feasible server chain as shown in Fig.~\ref{fig:example_cache_allocation}b, and the residual memory allows us to simultaneously run $c_1=5$ jobs on chain $k_1$, $c_2=5$ jobs on chain $k_2$, and $c_3=5$ jobs on a new chain $k_3=j_1 \to j_4 \to j_5$,  achieving a higher total service rate of 
\begin{align}
\nu = {5\over 3+5\epsilon} + {5\over 3+10\epsilon} + {5\over 3+12\epsilon}
\end{align}
under the same block placement. 
This example suggests the need to further optimize server chain composition after block placement. 
\fi

\emph{Algorithm:}
The main challenge is complexity. In addition to the NP-hardness of optimal cache allocation as stated in Theorem~\ref{thm:NP-hardness of BPCA}, the dimension of the solution space is also large. 
\if\thisismainpaper 0 
Specifically, let  $\mathcal{G}_{\bm{a},\bm{m}}=(\mathcal{J}_+, \mathcal{E}_{\bm{a},\bm{m}})$ denote a logical \emph{routing topology} under a given block placement $(\bm{a},\bm{m})$, formed by the extended server set $\mathcal{J}_+$ and the set $\mathcal{E}_{\bm{a},\bm{m}}$ of server pairs that can be traversed consecutively. Then each path in $\mathcal{G}_{\bm{a},\bm{m}}$ from the dummy server $j_0$ to the dummy server $j_{J+1}$ represents a valid server chain that traverses all the blocks in order, as illustrated by Fig.~\ref{fig:example_cache_allocation}b. 
\fi
Generally, the number of possible server chains under a given block placement is exponential in the number of servers, which causes severe challenges for cache allocation. Fortunately, a given load balancing policy may only assign jobs to a subset of all the possible server chains, which allows efficient cache allocation by only focusing on this subset.

\begin{algorithm}[tb]
\small
\SetKwInOut{Input}{input}\SetKwInOut{Output}{output}
\Input{Block placement $(\bm{a},\bm{m})$, residual server memory $(\widetilde{M}_j)_{j\in \mathcal{J}}$ measured in \#cache slots }
\Output{Server chains $\mathcal{K}$ and capacities $\bm{c}:=(c_k)_{k\in\mathcal{K}}$ }
$M^{(0)}_j\leftarrow\widetilde{M}_j$, $\forall j\in \mathcal{J}$\;
$\mathcal{E}^{(0)}\leftarrow \{(i,j)\in \mathcal{E}_{\bm{a},\bm{m}}:\: M^{(0)}_j\geq m_{ij}\}$\;
$\mathcal{K}\leftarrow \emptyset$, $l\leftarrow 1$\;
\While{$j_0$ is connected to $j_{J+1}$ in $\mathcal{G}^{(l-1)}:=(\mathcal{J}_+, \mathcal{E}^{(l-1)})$}
{
$k_l\leftarrow$ shortest $j_0\to j_{J+1}$ path in $\mathcal{G}^{(l-1)}$, with link cost $\tau^c_j+\tau^p_j m_{ij}$ for each $(i,j)\in \mathcal{E}^{(l-1)}$\; \label{GCA:4}
$\mathcal{K}\leftarrow \mathcal{K}\cup \{k_l\}$\;
$c_{k_l}\leftarrow \min_{(i,j)\in k_l} \lfloor {M^{(l-1)}_j\over m_{ij}} \rfloor$\; \label{GCA:6}
$M^{(l)}_j\leftarrow M^{(l-1)}_j - m_{ij}c_{k_l}$, $\forall j: (i,j)\in k_l$\; \label{GCA:7}
$\mathcal{E}^{(l)}\leftarrow \mathcal{E}^{(l-1)}$\;
\For{each $(i,j)\in k_l$ \label{GCA:9}}
{
\If{$M^{(l)}_j<m_{ij}$}
{$\mathcal{E}^{(l)}\leftarrow \mathcal{E}^{(l)}\setminus \{(i,j)\}$\;}
}\label{GCA:11}
$l\leftarrow l+1$\;
}
\caption{Greedy Cache Allocation (GCA)}
\vspace{-.0em}
\label{Alg:GCA}
\end{algorithm}
\normalsize

In particular, for load balancing among heterogeneous job servers with a central queue as considered in this work, a state-of-the-art load balancing policy is \emph{Join-the-Fastest-Free-Server (JFFS)}~\cite{Bhambay22PE}. Under JFFS, each incoming job will either be assigned to the fastest free job server if there is any, or entered into a central queue if none of the job servers is free. 
We find that this policy only utilizes a polynomial-sized subset of possible server chains, identified by \emph{Greedy Cache Allocation (GCA)} as shown in Alg.~\ref{Alg:GCA}. Let $M^{(l)}_j$ denote the residual memory at server $j$ after constructing and allocating cache for $l$ chains, measured in the number of cache slots; initially, $M^{(0)}_j=\widetilde{M}_j$ as defined in \eqref{eq:tildeM_j}.  Recall from Section~\ref{subsec:System Model} that $m_{ij}:=a_j+m_j-a_i-m_i$ is the number of cache slots occupied at server $j$ by each job processed by servers $i$ and $j$ consecutively. The algorithm sequentially allocates cache space by (i) finding the fastest chain among all the possible chains with available memory for at least one more job using shortest path routing (line~\ref{GCA:4}), (ii) allocating as much cache space as possible to this chain to maximize the number of jobs it can process concurrently (line~\ref{GCA:6}), and (iii) updating the residual server memory (line~\ref{GCA:7}) as well as the server pairs that can serve at least one more job  concurrently (lines~\ref{GCA:9}--\ref{GCA:11}). The process continues until no more job can be served using the remaining server memory. 
\if\thisismainpaper 0
For example, for the block placement in Fig.~\ref{fig:example_cache_allocation}b, GCA constructs $3$ server chains $\mathcal{K}=\{j_1\to j_2, j_1\to j_4 \to j_5, j_3\to j_4 \to j_5\}$ with a capacity of $5$ each. 
\fi

\emph{Complexity:} 
Each \textbf{while} loop in Alg.~\ref{Alg:GCA} will remove at least one link, as the link $(i^*,j^*)\in k_l$ achieving the $\min$ in line~\ref{GCA:6} must have $M^{(l)}_{j^*}<m_{i^* j^*}$ and thus be removed in line~\ref{GCA:11}. Thus, the number of \textbf{while} loops is bounded by $O(|\mathcal{E}^{(0)}|) = O(J^2)$. The time for each \textbf{while} loop is dominated by the time for shortest path routing in line~\ref{GCA:4}, which is $O(J\log{J}+J^2)$. Thus, the total time complexity of Alg.~\ref{Alg:GCA} is bounded by $O(J^4)$. This analysis also implies that the number of server chains constructed by GCA is bounded by $O(J^2)$, which can be much smaller than the total number of possible chains under the given block placement. 

\emph{Analysis:}
Even if the server chains constructed by GCA only constitute a small subset of all possible chains, they are sufficient for supporting JFFS type of job assignment as proved below. 

\begin{theorem}\label{thm:GCA conditional optimality}
Under any given block placement $(\bm{a},\bm{m})$, if GCA (Alg.~\ref{Alg:GCA}) returns a set of server chains $\mathcal{K}$ and the corresponding capacities $(c_k)_{k\in\mathcal{K}}$, then to assign jobs according to the principle of JFFS, it suffices to limit the utilized server chains to $\mathcal{K}$ and the number of concurrent jobs on each chain $k\in \mathcal{K}$ to $c_k$. 
\end{theorem}

\emph{Remark:} Theorem~\ref{thm:GCA conditional optimality} implies that instead of considering all the exponentially many possible server chains, we only need to consider the $O(J^2)$ server chains constructed by GCA and bound their capacities accordingly, if the load balancing policy is JFFS. 
Note that this holds under any block placement, not only the block placement given by GBP-CR.

\subsection{Load Balancing}\label{subsec:Load Balancing}

As explained in Section~\ref{subsec:Optimization Problems}, after server chain composition (via block placement and cache allocation), the remaining task for the orchestrator becomes a classical load balancing problem of assigning incoming jobs to the composed server chains under their capacity constraints. 
In particular, since our system has a central queue, we can apply the JFFS policy~\cite{Bhambay22PE}, which has demonstrated superior empirical performance in comparison to the state-of-the-art load balancing policies with dedicated queues~\cite{Ainbinder24AS} and been used to establish a lower bound on mean response time for stationary load balancing policies~\cite{Bhambay22PE}.

\subsubsection{Adaptation of JFFS}
Since the original JFFS does not consider parallel processing at servers, we need to adapt it to account for the capability for a server chain to process multiple jobs in parallel. The adapted policy is shown in Alg.~\ref{Alg:JFFC}, referred to as \emph{Join-the-Fastest-Free-Chain (JFFC)}. Here $Z_k(t)$ ($k\in \mathcal{K}$) denotes the number of ongoing jobs on chain $k$, and $Q(t)$ denotes the number of jobs waiting in the central queue, both at time $t$.

\begin{algorithm}[tb]
\small
\SetKwInOut{Input}{input}\SetKwInOut{Output}{output}
\Input{Server chains $\mathcal{K}$ with service rates $(\mu_k)_{k\in \mathcal{K}}$ and capacities $(c_k)_{k\in \mathcal{K}}$}
\Output{Assignment decision for each job}
\tcp{Upon the arrival of a new job at time $t$:}
\If{$\exists k\in \mathcal{K}$ with $Z_k(t)<c_k$}
{$k^*\leftarrow \argmax\{\mu_k:\: Z_k(t)<c_k\}$\;
assign the job to chain $k^*$\;
}
\Else{add the job to the end of the central queue\;}
\tcp{Upon the completion of an existing job on chain $k$:}
\If{$Q(t)>0$}
{assign the first queued job to chain $k$\;}
\caption{Join-the-Fastest-Free-Chain (JFFC)}
\vspace{-.0em}
\label{Alg:JFFC}
\end{algorithm}
\normalsize

\subsubsection{Response Time Analysis}\label{subsubsec:Response Time Analysis}
Our focus is on the steady-state mean response time. 
 For this analysis, we assume that:
 \begin{enumerate}
 \item Jobs arrive according to a Poisson process of rate $\lambda$, each requiring an independent and exponentially distributed amount of work with mean $1$ (i.e., the service times for a server chain with service rate $\mu$ are independent and exponentially distributed with mean $1/\mu$); 
 \item The service times are independent of the inter-arrival times. 
 \end{enumerate}
 We note that Poisson arrivals and exponential service times are standard assumptions for steady-state analysis of load balancing systems and not needed in our previous algorithm development or analysis. As shown later (Section~\ref{subsec:PETALS-based Experiments}), while real arrival processes and service times may deviate from these assumptions, analytical results derived from these assumptions still provide meaningful guidance in practice.

For the ease of presentation, we sort the server chains into descending order of service rates $\{k_l\}_{l\in [K]}$ ($K:=|\mathcal{K}|$), and use $\mu_l$/$c_l$ to denote the service rate/capacity of the $l$-th fastest chain $k_l$. 
We denote the system state at time $t$ by a vector $\bm{Z}(t)$, where $Z_0(t)$ is the \emph{number of queued jobs} and $Z_l(t)$ ($l\in [K]$) is the \emph{number of ongoing jobs on the $l$-th fastest chain}. 
Under the above assumptions, $(\bm{Z}(t))_{t\geq 0}$ is a \emph{continuous-time Markov chain (CTMC)}, with a state space $\mathcal{Z}$ consisting of two disjoint parts: 
\begin{align}
\mathcal{Z}_1 &:= \{\bm{z}:\: z_0 = 0 \mbox{ and } z_l\in \{0,\ldots,c_l\}, \forall l\in [K]\}, \\
\mathcal{Z}_2 &:= \{(n,c_1,\ldots,c_K):\: n\in \mbbZ^+\}.
\end{align}
Under JFFC, the transition rate from $\bm{z}$ to $\bm{z}'$ is given by\footnote{We use $\bm{x}_{i_1:i_2}$ to denote the vector $(x_i)_{i=i_1}^{i_2}$ of a sub-range of elements.} 
\begin{align}
\hspace{-.75em} q(\bm{z},\bm{z}')\hspace{-.25em} =\hspace{-.25em} \left\{\hspace{-.5em}\begin{array}{ll}
\lambda & \hspace{-2em}\mbox{if }\bm{z}\mbox{ transitions into }\bm{z}' \mbox{ upon an arrival},\\
z_l\mu_l &\hspace{-2em} \mbox{if }z'_0=z_0=0, z'_l = z_l-1, z'_{l'}=z_{l'} (\forall l'\neq l),\\
\sum_{l=1}^K c_l\mu_l &\mbox{if }z'_0=z_0-1,\bm{z}'_{1:K}=\bm{z}_{1:K}=\bm{c}_{1:K},\\
0 & \mbox{o.w.,}
\end{array}
\right. \label{eq:q(z,z')}
\end{align}
where $\bm{z}$ upon an arrival will transition into $\bm{z}'$ with $z'_0=z_0+1$ and $\bm{z}'_{1:K} = \bm{z}_{1:K}$ if $\bm{z}_{1:K} = \bm{c}_{1:K}$ (i.e., all the chains are fully occupied), or $z'_{l^*} = z_{l^*}+1$ if $\exists l^*=\argmin\{l:\: z_l<c_l\}$ and $z'_l=z_l$ for $l\neq l^*$ (i.e., $k_{l^*}$ is the fastest chain with available capacity). 

JFFC is throughput-optimal, i.e., it guarantees stability whenever the arrival rate is less than the total service rate, as stated below. 

\begin{lemma}\label{lem:ergodicity}
The CTMC $(\bm{Z}(t))_{t\geq 0}$ with transition rate \eqref{eq:q(z,z')} is ergodic for any $\lambda< \sum_{l=1}^K c_l \mu_l$. 
\end{lemma}

Lemma~\ref{lem:ergodicity} implies that as long as $\lambda< \sum_{l=1}^K c_l \mu_l$, the system will have a unique \emph{steady-state distribution, denoted by $\bm{\pi}$}, which characterizes the \emph{steady-state mean response time $\overline{T}$} by Little's law as \looseness=-1
\begin{align}\label{eq:steady-state mean response time}
\overline{T} = \mbbE_{\bm{\pi}}\left[\sum_{l=0}^K Z_l\right]/\lambda,
\end{align}
where $\bm{Z}\sim\bm{\pi}$ is the steady-state system state. 
Exact analysis of the steady-state distribution of $(\bm{Z}(t))_{t\geq 0}$ is highly challenging due to the high-dimensional nature of the state space; 
\if\thisismainpaper 1
see Appendix~A.3 in \cite{Sun26arXiv} for the exact analysis in a special case of $K=2$. 
\else
see Appendix~\ref{appendix:Analysis of JFFC for K=2} for the exact analysis in a special case of $K=2$. 
\fi 
Nevertheless, by collapsing the states with suitably defined transition rates, we can derive closed-form upper/lower bounds as follows. 

Define $\mathcal{Z}_n:=\{\bm{z}\in \mathcal{Z}:\: \sum_{l=0}^K z_l = n\}$ ($n\in \mbbN$) as the subset of states with $n$ jobs in the system. Then in the steady state, the flow balance equation between $\mathcal{Z}_{n-1}$ and $\mathcal{Z}_n$ is given by
\begin{align}
\lambda \sum_{\bm{z}\in \mathcal{Z}_{n-1}}\pi_{\bm{z}} = \sum_{\bm{z}\in \mathcal{Z}_n} \pi_{\bm{z}} \left(\sum_{l=1}^K z_l\mu_l\right) = \nu_n \sum_{\bm{z}\in\mathcal{Z}_n}\pi_{\bm{z}}, \label{eq:flow balance of Z_n}
\end{align}
where $\nu_n$ is the \emph{average steady-state death rate} for $\mathcal{Z}_n$, defined as\looseness=-1
\begin{align}\label{eq:nu_n}
\nu_n := \sum_{\bm{z}\in\mathcal{Z}_n}{\pi_{\bm{z}}\over \sum_{\bm{z}'\in\mathcal{Z}_n}\pi_{\bm{z}'}}\sum_{l=1}^K z_l \mu_l.
\end{align}
Therefore, if $(\nu_n)_{n\in\mbbN}$ is known, we can construct a birth-death process with a scalar state $\Phi(t)$ that represents the \emph{total number of jobs} in the system at time $t$, with a fixed birth rate $\lambda$ and a state-dependent death rate $\nu_n$ for each $n\in\mbbN$. The steady-state distribution $\bm{\phi}$ of this process must satisfy the flow balance equation $\lambda \phi_{n-1} = \nu_n \phi_n$, which compared to \eqref{eq:flow balance of Z_n} implies
$\phi_n = \sum_{\bm{z}\in\mathcal{Z}_n}\pi_{\bm{z}}$,
i.e.,\looseness=-1
\begin{align}
\mbbE_{\bm{\pi}}\left[\sum_{l=0}^K Z_l\right] = \sum_{n=0}^\infty n \sum_{\bm{z}\in\mathcal{Z}_n}\pi_{\bm{z}} = \sum_{n=0}^\infty n \phi_n = \mbbE_{\bm{\phi}}[\Phi]. \label{eq:equivalence of occupancy}
\end{align}
Thus, $(\Phi(t))_{t\geq 0}$ is equivalent to $(\bm{Z}(t))_{t\geq 0}$ in terms of mean system occupancy and hence mean response time in the steady state. 

However, we cannot construct $(\Phi(t))_{t\geq 0}$ as we do not know the death rate $\nu_n$ that depends on the unknown steady state distribution of $(\bm{Z}(t))_{t\geq 0}$ through \eqref{eq:nu_n}. Nevertheless, we can bound $\nu_n$ by\footnote{We use $(x)_+$ to denote $\max(x,\: 0)$.}
\begin{align}
\nu_n &\leq \sum_{l=1}^K \mu_l \cdot \min\left(c_l,\: \Big(n-\sum_{l'=1}^{l-1}c_{l'}\Big)_+ \right) =: \overline{\nu}_n, \label{eq:nu_n upper}\\
\nu_n &\geq \sum_{l=1}^K \mu_l \cdot \min\left(c_l,\: \Big(n-\sum_{l'=l+1}^K c_{l'} \Big)_+ \right) =: \underline{\nu}_n, \label{eq:nu_n lower}
\end{align}
where \eqref{eq:nu_n upper} is the maximum job departure rate under $n$ jobs achieved when the jobs run on the fastest chains, and \eqref{eq:nu_n lower} is the corresponding minimum job departure rate when the jobs run on the slowest chains. Accordingly, we can construct two birth-death processes: $(\underline{\Phi}(t))_{t\geq 0}$ with death rate $\overline{\nu}_n$ ($\forall n\in\mbbN$) and $(\overline{\Phi}(t))_{t\geq 0}$ with death rate $\underline{\nu}_n$ ($\forall n\in\mbbN$), both having a fixed birth rate $\lambda$. Their steady-state distributions help to bound the steady-state mean occupancy of our system as follows. 

\begin{theorem}\label{thm:mean occupancy bounds}
Let $\nu:=\sum_{l=1}^K c_l\mu_l$, $C:=\sum_{l=1}^K c_l$, $\rho:= \lambda/\nu$,  
\begin{align}\label{eq:phi_n}
\underline{\phi}_n := \left(1+\sum_{l=1}^{C-1}{\lambda^l\over \prod_{i=1}^l \overline{\nu}_i}+{\lambda^C \nu\over (\prod_{i=1}^C \overline{\nu}_i)(\nu-\lambda)} \right)^{-1} \prod_{i=1}^n{\lambda\over \overline{\nu}_i}
\end{align}
for $n=0,\ldots,C$, and define $\overline{\phi}_n$ similarly by replacing $\overline{\nu}_i$ with $\underline{\nu}_i$. For $\lambda<\nu$, the steady-state mean occupancy under JFFC is bounded by \looseness=0
\begin{align}
\mbbE_{\bm{\pi}}\left[\sum_{l=0}^K Z_l\right] &\geq \sum_{n=0}^{C-1}n \underline{\phi}_n + \underline{\phi}_C\left({\rho\over (1-\rho)^2}+{C\over 1-\rho} \right), \label{eq:E[Z] lower}\\
\mbbE_{\bm{\pi}}\left[\sum_{l=0}^K Z_l\right] &\leq \sum_{n=0}^{C-1}n \overline{\phi}_n + \overline{\phi}_C\left({\rho\over (1-\rho)^2}+{C\over 1-\rho} \right). \label{eq:E[Z] upper}
\end{align}
\end{theorem}

Plugging the bounds from Theorem~\ref{thm:mean occupancy bounds} into \eqref{eq:steady-state mean response time} then yields closed-form upper and lower bounds on the steady-state mean response time under JFFC. 

\emph{Remark:}
We note that the existing analysis of JFFS in \cite{Bhambay22PE} only modeled the system state by a scalar $Z(t):= \sum_{l=0}^K Z_l(t)$ (i.e., the {total number of jobs} in the system at time $t$), under the assumption that jobs will always occupy the fastest job servers under JFFS. In contrast, we show that this assumption is incorrect in the steady state 
\if\thisismainpaper 1
(see the example in Appendix~A.3 in \cite{Sun26arXiv}), 
\else
(see the example in Appendix~\ref{appendix:Analysis of JFFC for K=2}), 
\fi
and only leads to a lower bound as shown by Theorem~\ref{thm:mean occupancy bounds}. 

\subsubsection{Improved Parameter Tuning}
The closed-form bounds in Theorem~\ref{thm:mean occupancy bounds} enable another method of tuning parameter $c$ in GBP-CR (Alg.~\ref{Alg:GBP-CR}): instead of minimizing the surrogate objective in \eqref{eq:optimal c}, we can choose $c$ to minimize the upper/lower bound in Theorem~\ref{thm:mean occupancy bounds} for the chains constructed by GBP-CR (Alg.~\ref{Alg:GBP-CR}) and GCA (Alg.~\ref{Alg:GCA}). Compared to the surrogate objective, these bounds better represent the actual mean response time under the proposed solution, which leads to better parameter tuning 
\if\thisismainpaper 1
(see Fig.~6 in \cite{Sun26arXiv}), 
\else
(see Fig.~\ref{fig:parameter_optimization}), 
\fi 
at the cost of slightly more computation at the orchestrator (because instead of only running GBP-CR for each candidate value of $c$, the orchestrator now needs to run both GBP-CR and GCA).


\section{Performance Evaluation}\label{sec:Performance Evaluation}

We will test our proposed solutions against benchmarks in serving LLM inference requests by distributed GPU nodes via two complementary evaluation methods: (i) model-driven simulations that validate our solution under its original assumptions, and (ii) experiments that evaluate the actual performance when applying our solution to a leading open-source Internet‑distributed LLM inference system called PETALS~\cite{Borzunov23NeurIPS}. 

\subsection{Model-driven Simulations}\label{subsec:Model-driven Simulations}

\subsubsection{Simulation Settings}\label{subsubsec:Simulation Settings}


\if\thisismainpaper 1
We simulate the processing of chain-structured jobs according to the system model in Section~\ref{subsec:System Model}, with parameters set to represent a realistic case of pipeline-parallel LLM inference, assuming a 176B-parameter LLM and two types of servers with GPUs comparable to 3g.40gb and 2g.20gb generated from A100 by Nvidia's multi-instance GPU technology~\cite{NvidiaMIG}\footnote{Our simulation code is available at: \url{https://github.com/TingyangSunJeff/load_balance_llm_simulator}.}.     
We consider a geographically-distributed deployment as assumed by PETALS~\cite{Borzunov23NeurIPS}, and set $\tau^c_j$ 
according to communication delays in the RIPE Atlas European network~\cite{ripe_atlas}. See details in Section~4.1.1 of \cite{Sun26arXiv}. 
We randomly select one node as the orchestrator location and $J$ other nodes as server locations, where $\eta$ fraction of randomly selected servers are equipped with high-performance GPUs and the rest are equipped with low-performance GPUs.  
Jobs arrive according to a Poisson process of rate $\lambda$, each having an independent and exponentially distributed size with mean $1$, where serving a job of size $r$ by a server chain with service rate $\mu_k$ takes time $r/\mu_k$. We set $\lambda=0.2 ~\text{requests/s}$ and $\overline{\rho}=0.7$. Results are averaged over 20 Monte Carlo runs.  
\else
We simulate the processing of chain-structured jobs according to the system model in Section~\ref{subsec:System Model}, with parameters set to represent the use case of pipeline-parallel LLM inference\footnote{Our simulation code is available at: \url{https://github.com/TingyangSunJeff/load_balance_llm_simulator}.}. Specifically, we set $s_m=1.32$ GB and $s_c=0.11$ GB according to the tensor sizes for model parameters and KV cache\footnote{The KV cache size is computed based on a maximum sequence length of 2048.} per transformer layer and $L=70$ according to the number of layers for BLOOM-176B~\cite{BigScience23BLOOM}. Assuming an average input length of $2,000$ tokens and an average output length of $20$ tokens according to the traces from \cite{Patel2024Splitwise}, we set $M_j$ and $\tau^p_j$ according to two types of GPUs: $M_j= 40$ GB and $\tau^p_j=109$ ms\footnote{\label{fn:tau^p model}Let $t_o\approx 1$ ms denote the per-block overhead time, $t^I_j$ denote the per-block-per-token prefill time, $t^O_j$ denote the per-block-per-token decode time, and $\overline{l}_{in}$/$\overline{l}_{out}$ denote the average input/output length. Then $\tau^p_j=t_o + t^I_j\overline{l}_{in} + t^O_j (\overline{l}_{out}-1)$. BLOOM-176B requires $F=5$ GFLOPs per block per token and $s_m=1.32$ GB per block under NF4, which means that for a server with $f_j$ TFLOPS and $b_j$ GB/ms memory bandwidth, $t^I_j\approx F/f_j$ ms and $t^O_j\approx s_m/b_j$ ms, as prefill is compute-bound and decode is memory-bound.} 
for a high-performance GPU node with 120 TFLOPS and a memory bandwidth of $1.02$ GB/ms, and $M_j=20$ GB and $\tau^p_j=175$ ms for a low-performance GPU node with 80 TFLOPS and a memory bandwidth of $0.51$ GB/ms. These GPU types are comparable to slices of A100 (3g.40gb and 2g.20gb) generated by Nvidia's multi-instance GPU technology~\cite{NvidiaMIG}. 

We consider geographically-distributed deployments as targeted by PETALS~\cite{Borzunov23NeurIPS}, and set $\tau^c_j$ 
according to the communication delays in a real wide-area network, 
the RIPE Atlas European network~\cite{ripe_atlas}. 
To simulate the communication method in PETALS~\cite{Borzunov23NeurIPS}, where the node initiating an inference session (i.e., the orchestrator) relays data between servers, we set $\tau^c_j$ according to the RTT between the nodes hosting the orchestrator and server $j$, 
plus $18$ ms to account for overhead (e.g.,  serialization or deserialization time). We randomly select one node as the orchestrator location and $J$ other nodes as server locations, where $\eta$ fraction of randomly selected servers are equipped with high-performance GPUs and the rest are equipped with low-performance GPUs.  
Jobs arrive according to a Poisson process of rate $\lambda$, each having an independent and exponentially distributed size with mean $1$, where serving a job of size $r$ by a server chain with service rate $\mu_k$ takes time $r/\mu_k$. By default, $J = 20$, $\eta = 0.2$, $\lambda=0.2 ~\text{requests/s}$, and $\overline{\rho}=0.7$. Results are averaged over 20 Monte Carlo runs. 
\fi


\subsubsection{Unit Tests}
We first validate the performance of individual algorithms under their design assumptions. 
\if\thisismainpaper 1
The results, shown in Section~4.1.2 of \cite{Sun26arXiv}, not only validate our theoretical analysis of the proposed algorithms (Alg.~\ref{Alg:GBP-CR}--\ref{Alg:JFFC}), but also demonstrate the importance of properly configuring the cache reservation parameter $c$ and the efficacy of our bounds from Theorem~\ref{thm:mean occupancy bounds} (particularly the lower bound) in producing a good configuration. 
\fi

\if\thisismainpaper 0
\emph{Block placement:}
To evaluate GBP-CR (Alg.~\ref{Alg:GBP-CR}) in solving \eqref{eq:BP}, we compare the objective in \eqref{BP:obj} scaled by $c$, which is a surrogate for the mean response time, between GBP-CR and $100$ randomly generated feasible block placements 
reported as a box–whisker plot in Fig.~\ref{fig:gbp_cr_boxplots}. To validate Theorem~\ref{thm:optimality of GBP-CR} (optimality of GBP-CR under homogeneous memory), we separately test a homogeneous case (Fig.~\ref{fig:gbp_cr_boxplot_homogeneous}) and a heterogeneous case (Fig.~\ref{fig:gbp_cr_boxplot_heterogeneous}). In both cases, GBP-CR achieves equally good or better performance compared to the best solution found by randomized brute-force search. 


\begin{figure}[!t]
    \centering
    \begin{subfigure}[t]{0.4\linewidth}
        \centering
        \includegraphics[width=\linewidth]{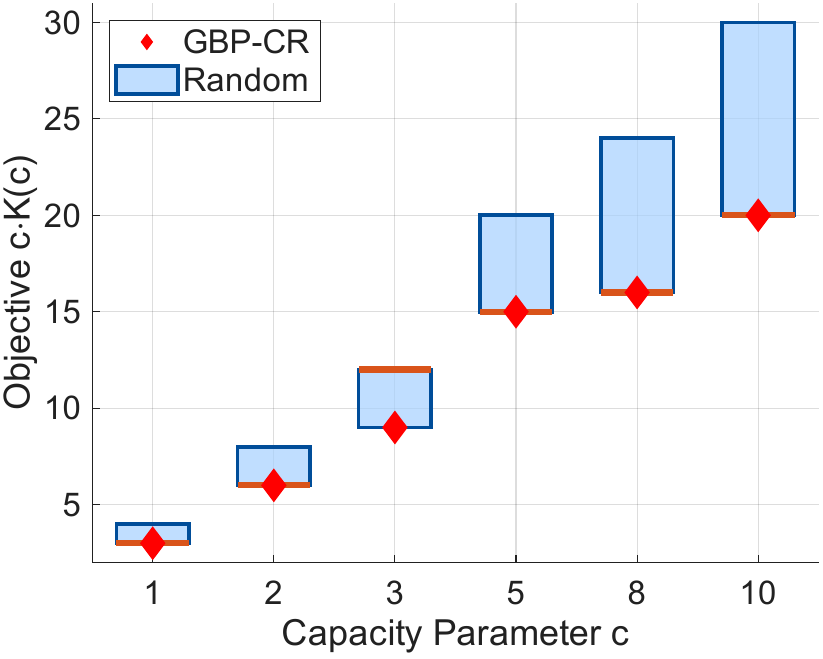}
        \caption{(a) $\eta = 0$}
        \label{fig:gbp_cr_boxplot_homogeneous}
    \end{subfigure}
    \begin{subfigure}[t]{0.4\linewidth}
        \centering
        \includegraphics[width=\linewidth]{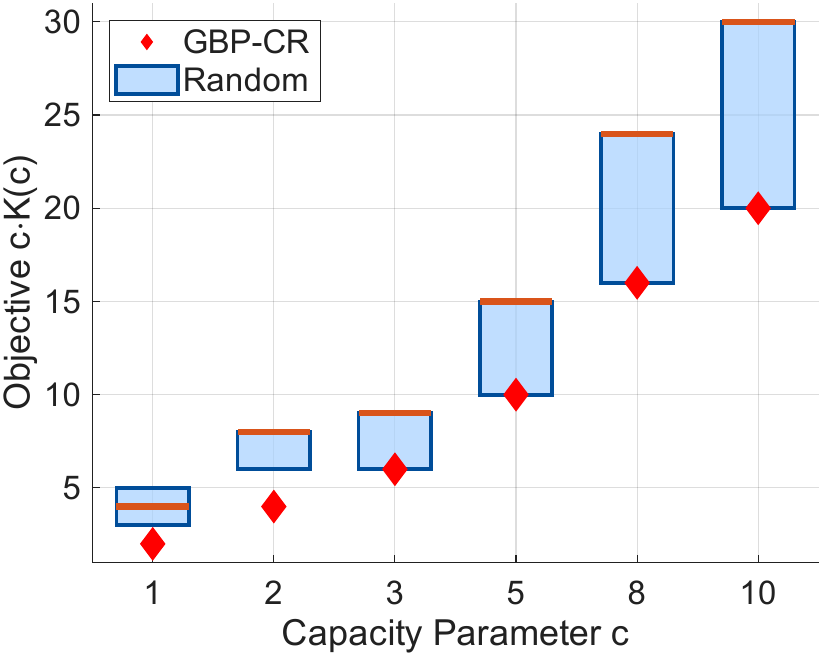}
        \caption{(b) $\eta = 0.2$}
        \label{fig:gbp_cr_boxplot_heterogeneous}
    \end{subfigure}
    \caption{Evaluation of GBP-CR under (a) homogeneous settings and (b) heterogeneous settings ($c=7$).}
    \label{fig:gbp_cr_boxplots}
\end{figure}



\emph{Cache allocation:}
To evaluate GCA (Alg.~\ref{Alg:GCA}) in solving \eqref{eq:BPCA - conceptual} under a given block placement, we assess the number of ``job servers'' to achieve a required total service rate of $\lambda/\overline{\rho}$ as in \eqref{BPCA:obj} (the smaller, the better). Under a block placement given by GBP-CR, we compare the performance of GCA with (i) an upper bound yielded by  only using the disjoint server chains and the reserved cache space given by GBP-CR (`$c\cdot K(c)$'); (ii) a lower bound given by $\lceil \lambda/(\overline{\rho} \mu_1) \rceil$, where $\mu_1$ is the highest service rate per chain (`Lower Bound'); (iii) a conditionally optimal cache allocation by solving the integer linear program resulting from plugging the server chains constructed by GCA into \eqref{eq:BPCA - conceptual} (`Optimal ILP'). The result in Fig.~\ref{fig:gca_comparison} shows that further optimizing cache allocation by GCA can substantially improve the performance and even achieve optimality under light loads. 
%

\begin{figure}[!t]
    \centering
    \includegraphics[width=0.4\linewidth]{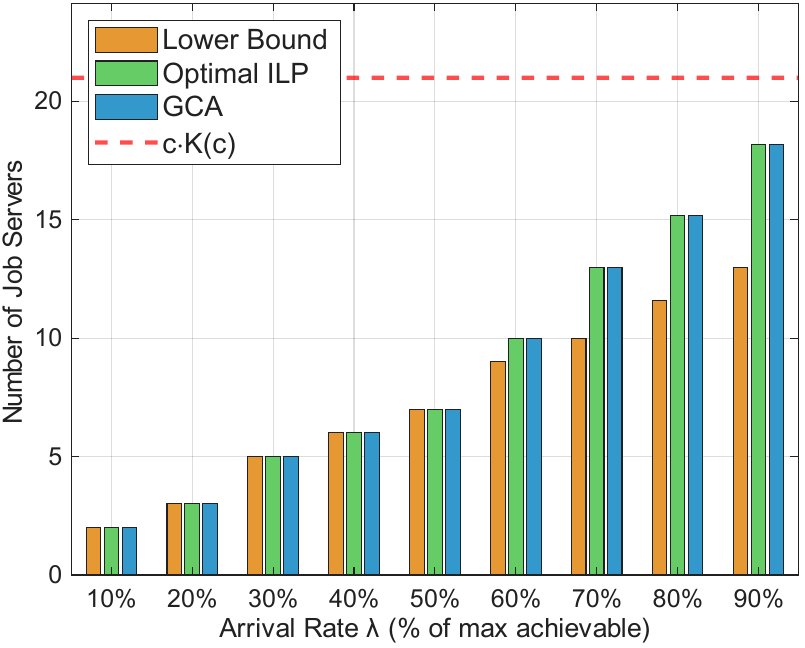}
    \caption{Evaluation of GCA under a fixed block placement produced by GBP-CR ($c=7$, $\lambda$ given in percentage of the total service rate of all the chains constructed by GCA). 
    }
    \label{fig:gca_comparison}
\end{figure}

\emph{Load balancing:}
We evaluate JFFC (Alg.~\ref{Alg:JFFC}) under a fixed set of server chains and their capacities given by GBP-CR + GCA in terms of the mean response time. Fig.~\ref{fig:jffc_other_policies} shows the comparison with other load balancing policies, including baselines like Join-the-Shortest-Queue (JSQ)~\cite{Winston77JAP} and Joint-the-Idle-Queue (JIQ)~\cite{Lu11PE} as well as more advanced policies like Smallest-Expected-Delay (SED) and Speed-Aware JSQ (SA-JSQ)~\cite{Bhambay22PE} (all the policies have been extended to account for parallel processing). Fig.~\ref{fig:jffc_bounds} shows the comparison with the lower/upper bound from Theorem~\ref{thm:mean occupancy bounds}. The results not only validate the state-of-the-art performance of JFFC, but also illustrate the gap between its actual performance and the bound that has been ignored in \cite{Bhambay22PE}.  
Moreover, we also observe that a large percentage (over $85\%$) of the mean response time for JFFC is the mean service time when the load factor is below $0.7$, which justifies our selection of $\overline{\rho}$. 

\begin{figure}[!t]
    \centering
    \begin{subfigure}[t]{0.4\linewidth}
        \centering
        \includegraphics[width=\linewidth]{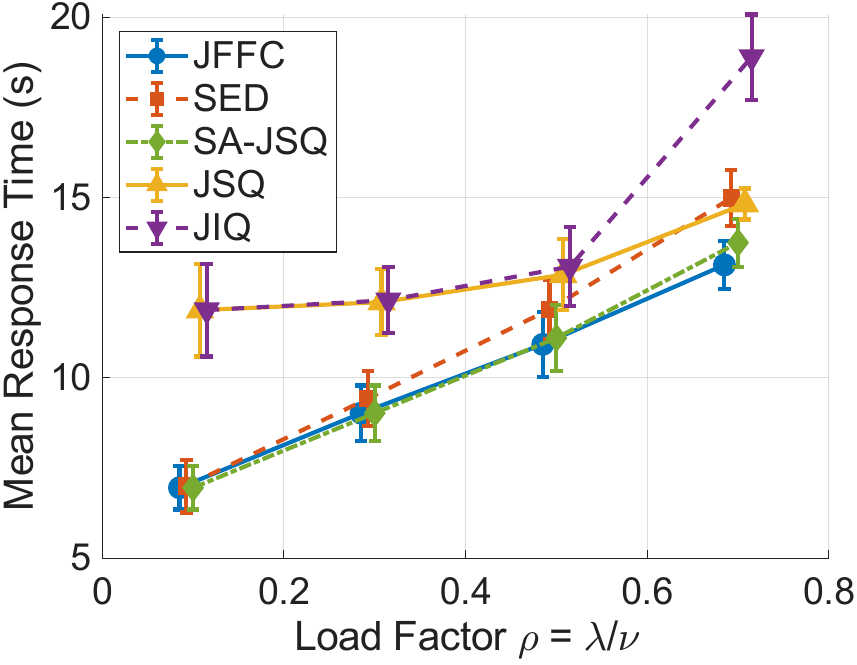}
        \caption{(a) comparison across policies}
        \label{fig:jffc_other_policies}
    \end{subfigure}
    \begin{subfigure}[t]{0.4\linewidth}
        \centering
        \includegraphics[width=\linewidth]{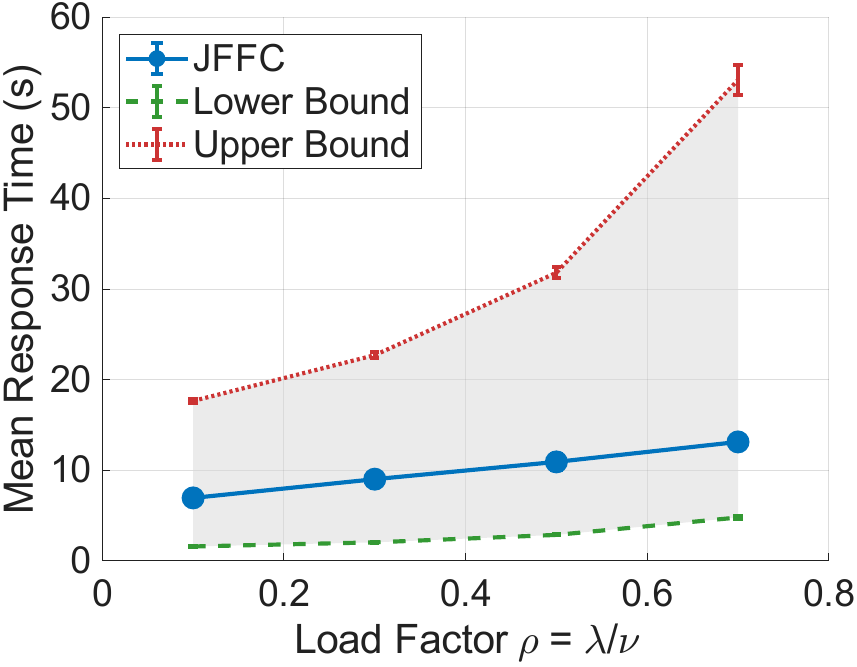}
        \caption{(b) comparison with bounds}
        \label{fig:jffc_bounds}
    \end{subfigure}
    \caption{Evaluation of JFFC under a fixed set of server chains produced by GBP-CR + GCA (for $c=7$).  }
    \label{fig:jffc_response_time_comparison}
\end{figure}


\begin{figure}[!t]
    \centering
    \begin{subfigure}[t]{0.4\linewidth}
        \centering
        \includegraphics[width=\linewidth]{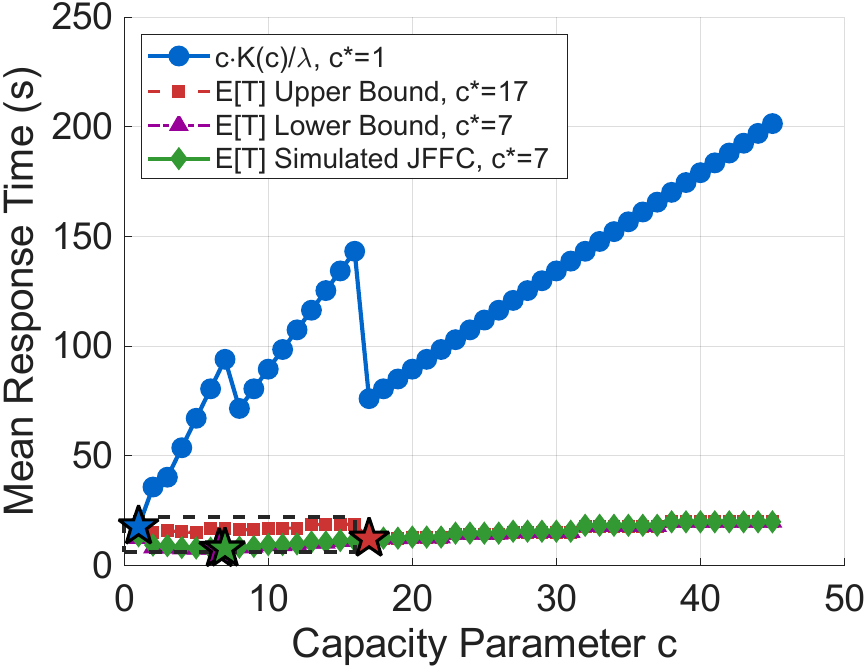}
        \caption{(a) overall view}
        \label{fig:parameter_optimization_c}
    \end{subfigure}
    \begin{subfigure}[t]{0.4\linewidth}
        \centering
        \includegraphics[width=\linewidth]{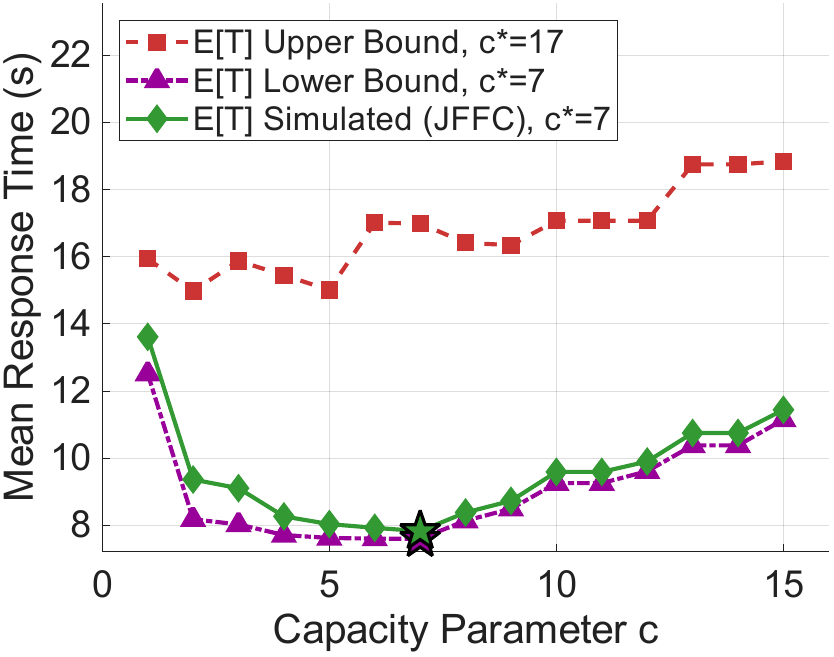}
        \caption{(b) zoomed-in view}
        \label{fig:parameter_optimization_c_zoom}
    \end{subfigure}
    \caption{Optimizing design parameter \(c\) (marked by $\star$). }
    \label{fig:parameter_optimization}
\end{figure}

\begin{figure}[t]
    \centering
    \includegraphics[width=0.4\textwidth]{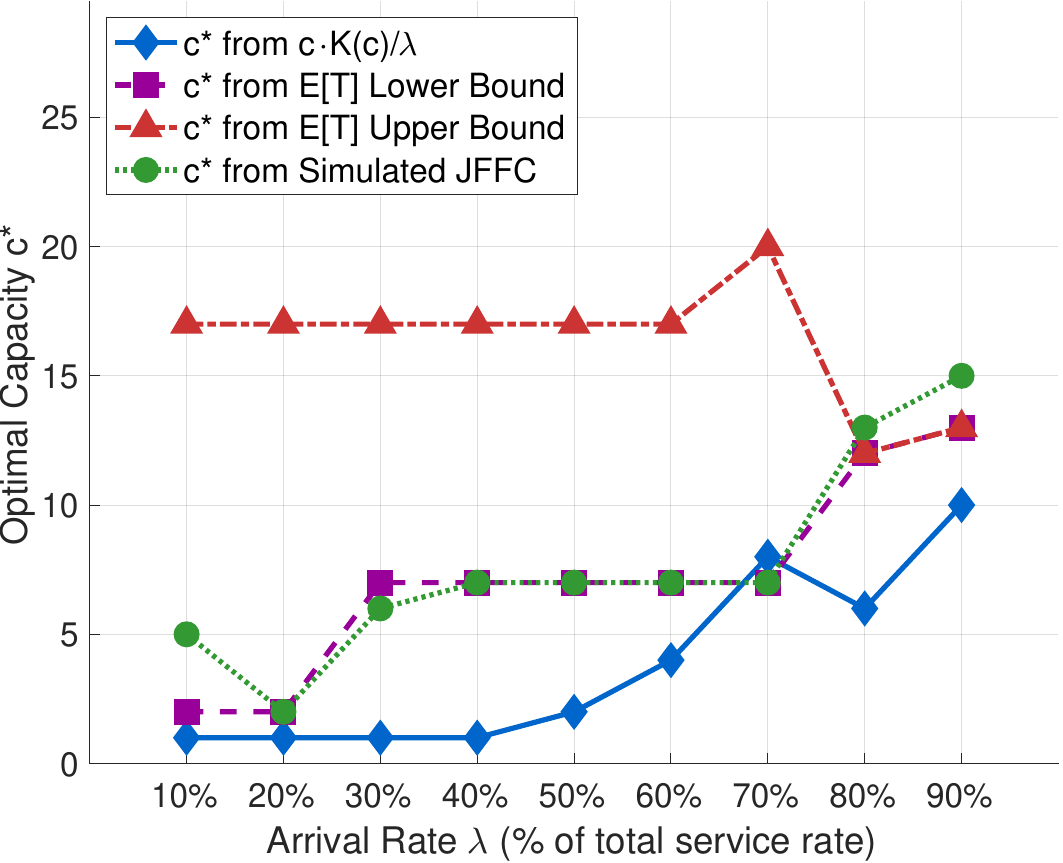}
    \caption{Optimal design parameter $c^*$ as a function of arrival rate $\lambda$. }
    \label{fig:optimal_c_vs_lambda}
\end{figure}

\emph{Parameter tuning:}
Fig.~\ref{fig:parameter_optimization} shows the impact of the design parameter $c$ in GBP-CR, which controls cache reservation, on the mean response time. In addition to the actual mean response time from simulations, we also evaluate: (i) the value of $c \cdot K(c)/\lambda$ that represents the objective of \eqref{eq:optimal c} (where dividing by $\lambda$ converts its meaning to time) and (ii) the upper/lower bound from Theorem~\ref{thm:mean occupancy bounds}. All the curves demonstrate a significant impact of $c$, 
which is not monotone. For example, the value of $c \cdot K(c)/\lambda$ grows linearly with $c$ most of the time, but when $c$ increases beyond a threshold, the number of disjoint chains $K(c)$ will decrease by one due to having fewer blocks per server, which will cause $c \cdot K(c)/\lambda$ to decrease by roughly $c/\lambda$. 
While the objective of \eqref{eq:optimal c} leads to too little cache reservation and the upper bound from Theorem~\ref{thm:mean occupancy bounds} leads to too much cache reservation, the lower bound from Theorem~\ref{thm:mean occupancy bounds} yields a configuration $c^*$ that minimizes the mean response time by reserving the right amount of cache space for each placed block. 

Fig.~\ref{fig:optimal_c_vs_lambda}  further evaluates how the optimal $c^*$ varies with the arrival rate $\lambda$. The lower bound from Theorem~\ref{thm:mean occupancy bounds} exhibits a monotone trend: at low arrival rates, it suggests smaller $c^*$ values, resulting in more memory being allocated to block placement to form shorter server chains and minimize service times; as the arrival rate increases, its suggested $c^*$ increases to allocate more memory to caches, enabling higher parallelism to handle the increased load. While simulation results show some variance due to randomness in the request process, they generally align with the lower bound's suggestion, validating its effectiveness for parameter tuning across different demand levels. 
In contrast, the $c^*$ value derived from the upper bound in Theorem~\ref{thm:mean occupancy bounds} is overly aggressive and the $c^*$ value derived from the surrogate objective $c\cdot K(c)/\lambda$ is overly conservative. \looseness=-1

\fi

\subsubsection{Overall Comparison}\label{subsubsec:Overall Comparison - simulation}

We compare the overall performance in terms of mean response time of the proposed solution, GBP-CR (for $c=7$) + GCA + JFFC, against two state-of-the-art solutions:
\begin{itemize} 
\item `PETALS': the current resource allocation algorithms used in the PETALS system~\cite{Borzunov23NeurIPS}, which greedily place blocks and route inference requests according to heuristic metrics;
\item `BPRR': a recently proposed solution from \cite{Sun25Performance}, which employs a two-time-scale algorithm to place blocks and dynamically route requests without explicitly composing server chains or allocating cache space ahead of time.
\end{itemize}

The results in Fig.~\ref{fig:overall_comparison_vs_J} under various server configurations show that: (i) by explicitly optimizing the composition of server chains, the proposed solution can significantly improve the performance over the state of the art (with 8\% -- 83\% reduction in mean response time), and (ii) the improvement is more prominent in more resource-constrained environments with fewer servers (smaller $J$) or a smaller fraction of high-performance servers (smaller $\eta$).

\begin{figure}[!t]
    \centering
    \begin{subfigure}{0.4\linewidth}
        \centering
        \includegraphics[width=\linewidth]{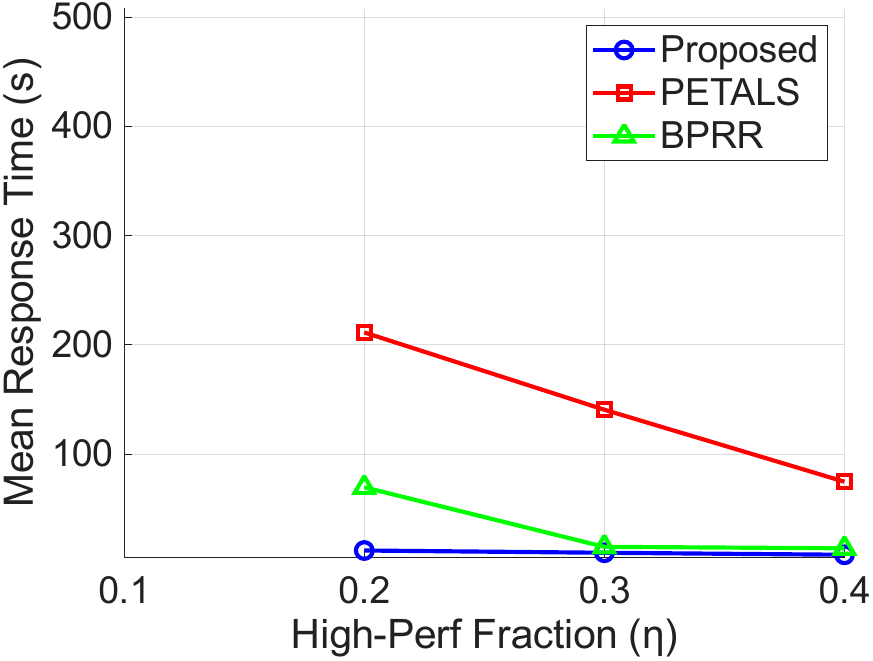}
        \caption{(a) $J=10$}
    \end{subfigure}
    \begin{subfigure}{0.4\linewidth}
        \centering
        \includegraphics[width=\linewidth]{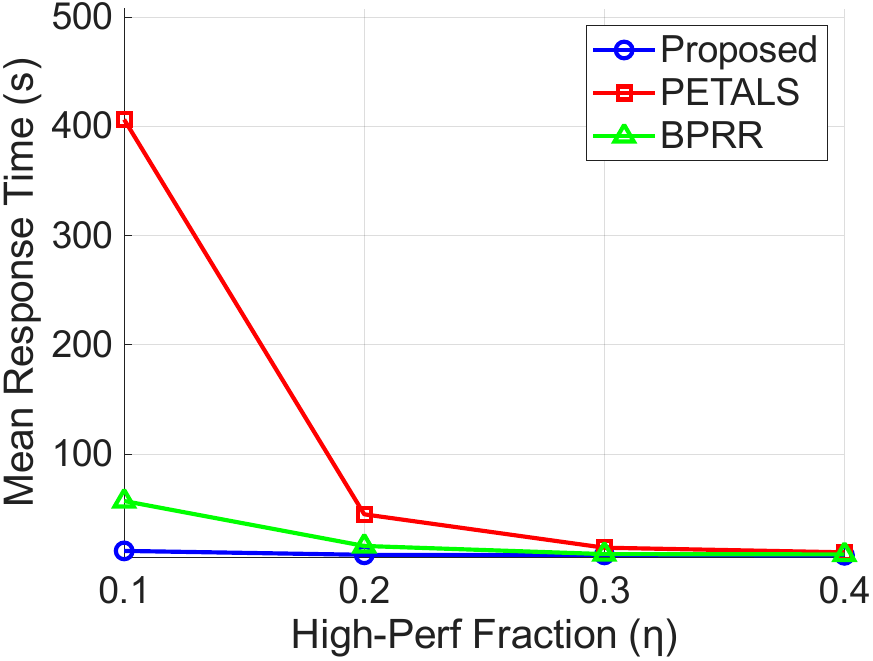}
        \caption{(b) $J=20$}
    \end{subfigure}
    \begin{subfigure}{0.4\linewidth}
        \centering
        \includegraphics[width=\linewidth]{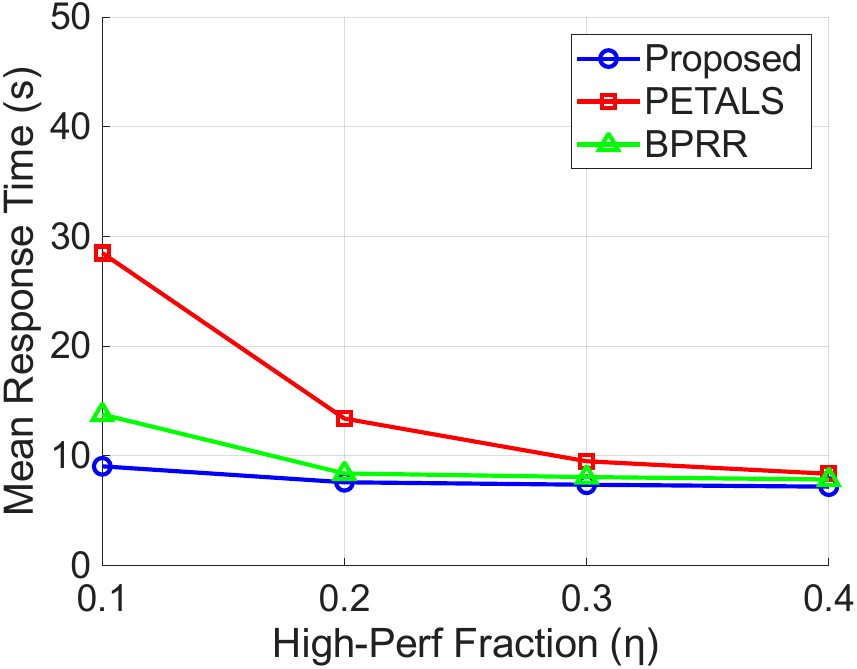}
        \caption{(c) $J=30$}
    \end{subfigure}
    \begin{subfigure}{0.4\linewidth}
        \centering
        \includegraphics[width=\linewidth]{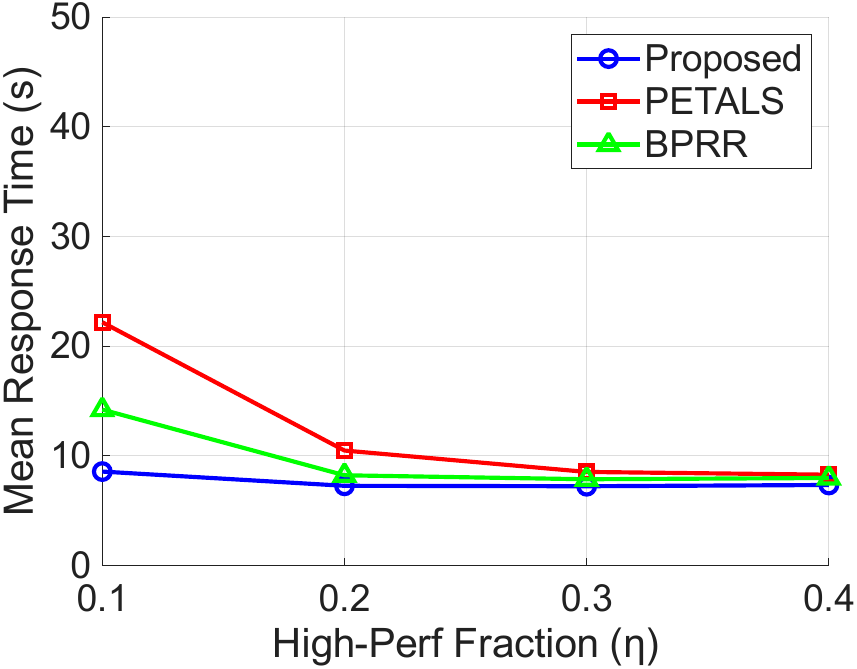}
        \caption{(d) $J=40$}
    \end{subfigure}
    \caption{Overall performance comparison with state-of-the-art solutions (the case of $J=10$ and $\eta=0.1$ is omitted as not all the $L$ blocks can be placed in this case). }
    \label{fig:overall_comparison_vs_J}
\end{figure}


\subsection{PETALS-based Experiments}\label{subsec:PETALS-based Experiments}

\subsubsection{Experiment Settings}\label{subsubsec:Experiment Settings}

We modify PETALS~\cite{Borzunov23NeurIPS} to implement the proposed server chain composition and load balancing algorithms\footnote{The modified PETALS code is available at: \url{https://github.com/TingyangSunJeff/petals_llm_chain_scheduling}.}. We deploy the system onto a server with 3 A100 (80 GB) GPUs, where we leverage the multi-instance GPU technology~\cite{NvidiaMIG} to partition each A100 into 2 lower-performance GPUs (2g.20gb) and 1 higher-performance GPU (3g.40gb) that together emulate 9 servers with one GPU each. Due to the limited GPUs, we use a smaller model LLaMA-2-7B for the experiment. 
We use the namespace and the traffic control features of Linux~\cite{schubert2019network} to simulate network latency according to the RTTs in RIPE Atlas European network~\cite{ripe_atlas} as described in Section~\ref{subsubsec:Simulation Settings}. As the measured RTTs only specify a fixed delay per link, we additionally consider a jitter configuration, where \texttt{netem} perturbs each link's delay by a normal distribution with a standard deviation of $10\%$ of the link's one-way delay.
We generate requests according to the Azure LLM inference trace from \cite{Patel2024Splitwise}, which contains the arrival times and the numbers of input/output tokens from sampled LLM inference services in Azure, collected on November 11th 2023.  The trace has an average rate of $2.57$ requests/s, an average input length of 2048 tokens, and an average output length of 28 tokens.  


\subsubsection{Model Validation}

\if\thisismainpaper 0
We start by estimating the parameters in our system model (e.g., $\tau^p_j$, $\tau^c_j$) and validating the resulting model with measurements from our PETALS deployment. Fig.~\ref{fig:computation time profiling} shows the computation time at a given server, decomposed into prefill time and decode time (of all generated tokens), as a function of the number of processed blocks, the output length, and the input length. The results validate that the overall computation time grows linearly with the number of processed blocks as assumed in \eqref{eq:mean service time model}. Moreover, the decode time grows linearly with the output length and the prefill time grows linearly with the input length, as assumed in footnote~\ref{fn:tau^p model}. 
Fig.~\ref{fig:comm_time_profiling} shows the communication time in invoking a given server, decomposed into communication time in generating the first token (during prefill) and communication time in generating the remaining tokens (during decode). This time is independent of the GPU type or the number of blocks processed at the server. The post-first-token communication time is strongly dependent on the output length due to autoregressive decoding. The first-token communication time weakly depends on the input length (due to data transfer time). 
In all cases, our system model closely matches the actual measurement. 

\begin{figure}[!t]
\centering
\begin{subfigure}[t]{0.4\linewidth}
\centering
\includegraphics[width=\linewidth]{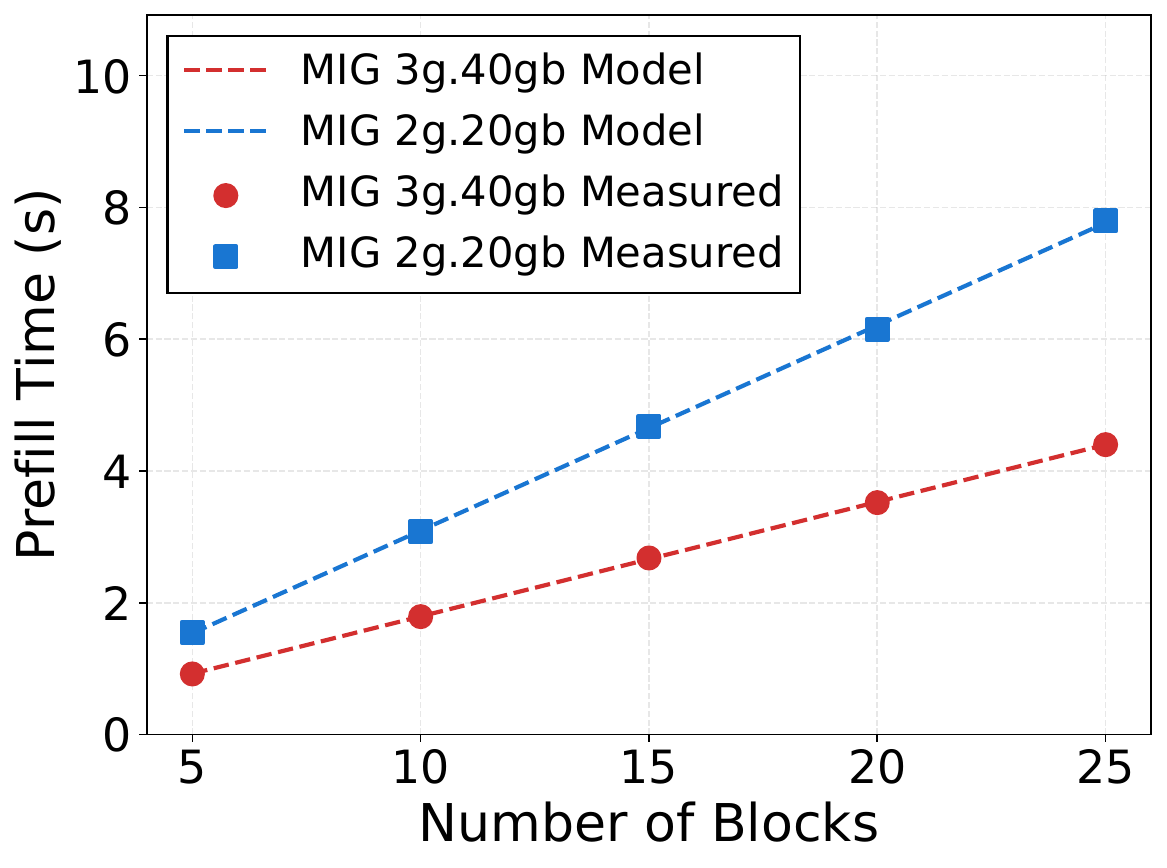}
\vspace{-1.75em}
\caption{(a) Prefill time vs. \#blocks}
\label{fig:prefill_blocks}
\end{subfigure}
\begin{subfigure}[t]{0.4\linewidth}
\centering
\includegraphics[width=\linewidth]{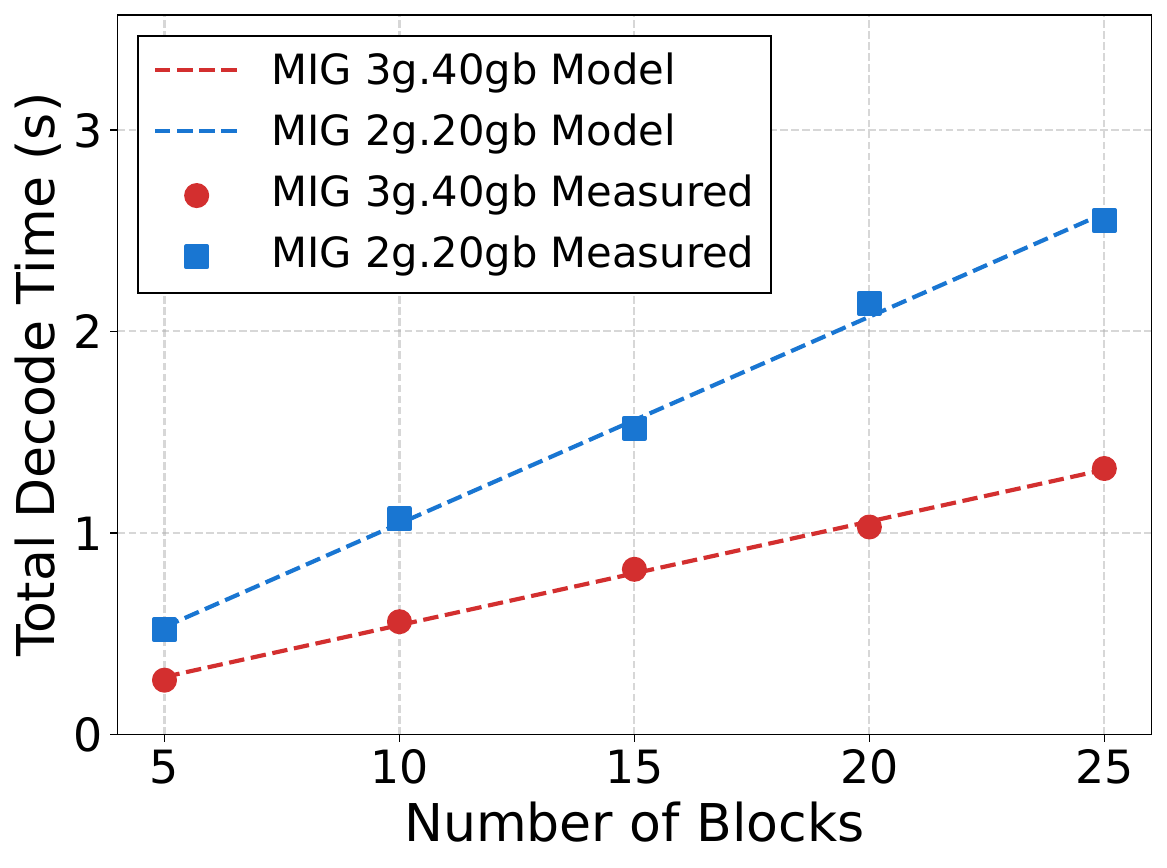}
\vspace{-1.75em}
\caption{(b) Decode time vs. \#blocks}
\label{fig:decode_blocks}
\end{subfigure}
\begin{subfigure}[t]{0.4\linewidth}
\centering
\includegraphics[width=\linewidth]{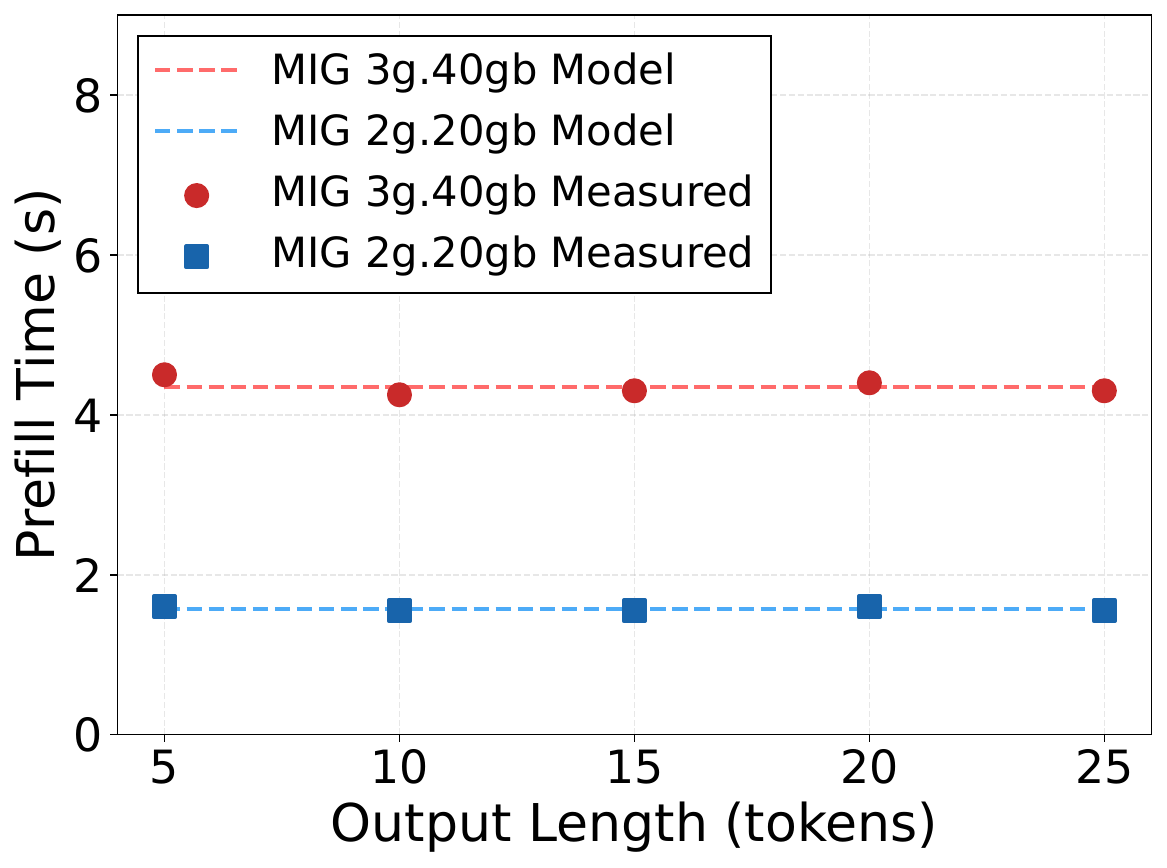}
\vspace{-1.75em}
\caption{(c) Prefill time vs. output length}
\label{fig:prefill_varying_output}
\end{subfigure}
\begin{subfigure}[t]{0.4\linewidth}
\centering
\includegraphics[width=\linewidth]{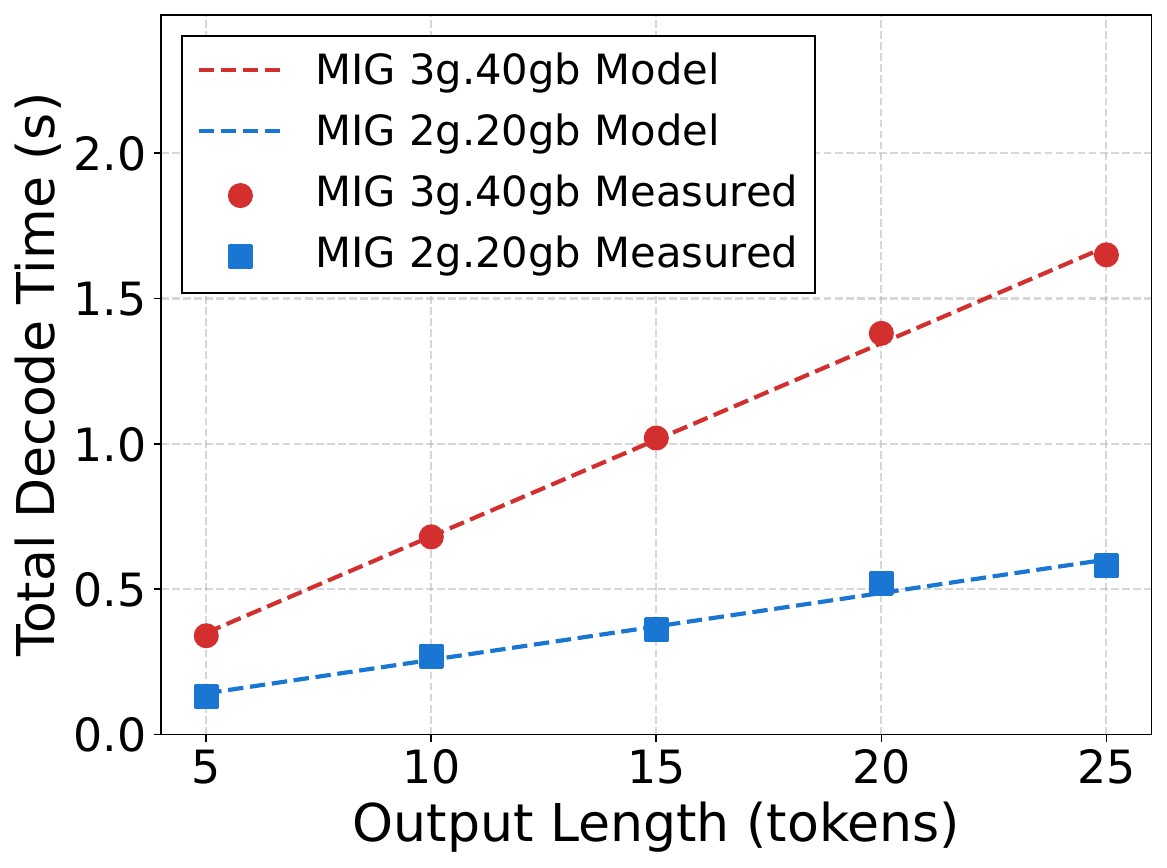}
\vspace{-1.75em}
\caption{(d) Decode time vs. output length}
\label{fig:decode_varying_output}
\end{subfigure}
\begin{subfigure}[t]{0.4\linewidth}
\centering
\includegraphics[width=\linewidth]{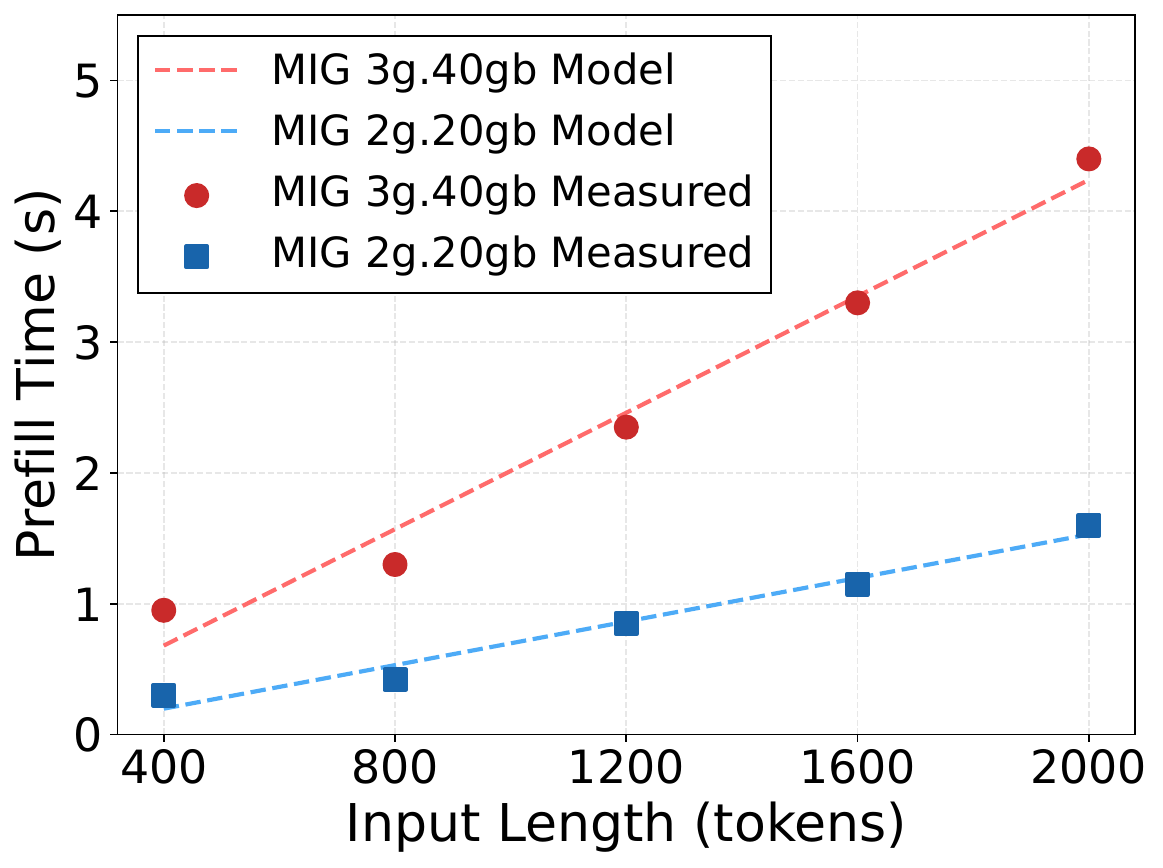}
\vspace{-1.75em}
\caption{(e) Prefill time vs. input length}
\label{fig:prefill_varying_input}
\end{subfigure}
\begin{subfigure}[t]{0.4\linewidth}
\centering
\includegraphics[width=\linewidth]{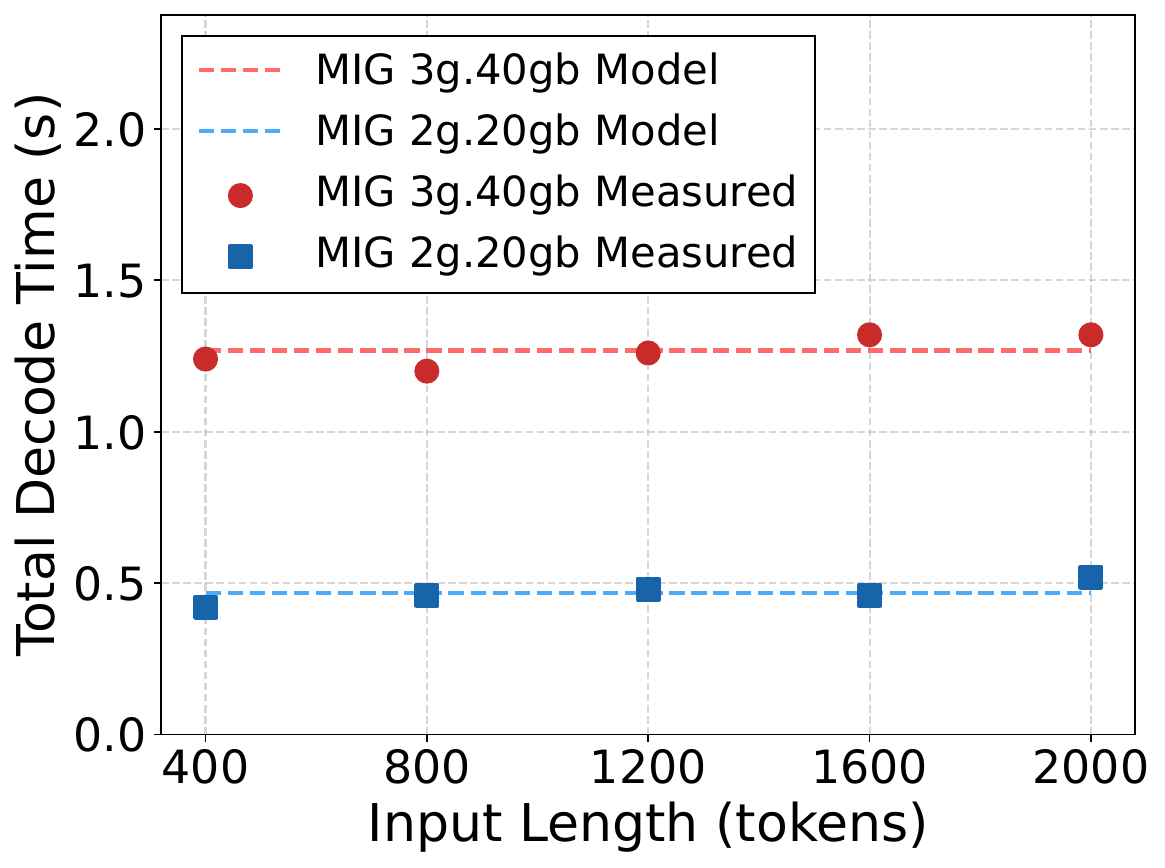}
\vspace{-1.75em}
\caption{(f) Decode time vs. input length}
\label{fig:decode_varying_input}
\end{subfigure}
\vspace{-0.75em}
\caption{Computation time profiling (by default, input length is 2,000 and output length is 20; MIG~3g hosts 25 blocks, MIG~2g hosts 5 blocks). }
\label{fig:computation time profiling}
\vspace{-0.25em}
\end{figure}

\begin{figure}[!t]
\centering
\begin{subfigure}[t]{0.4\linewidth}
    \centering
    \includegraphics[width=\linewidth]{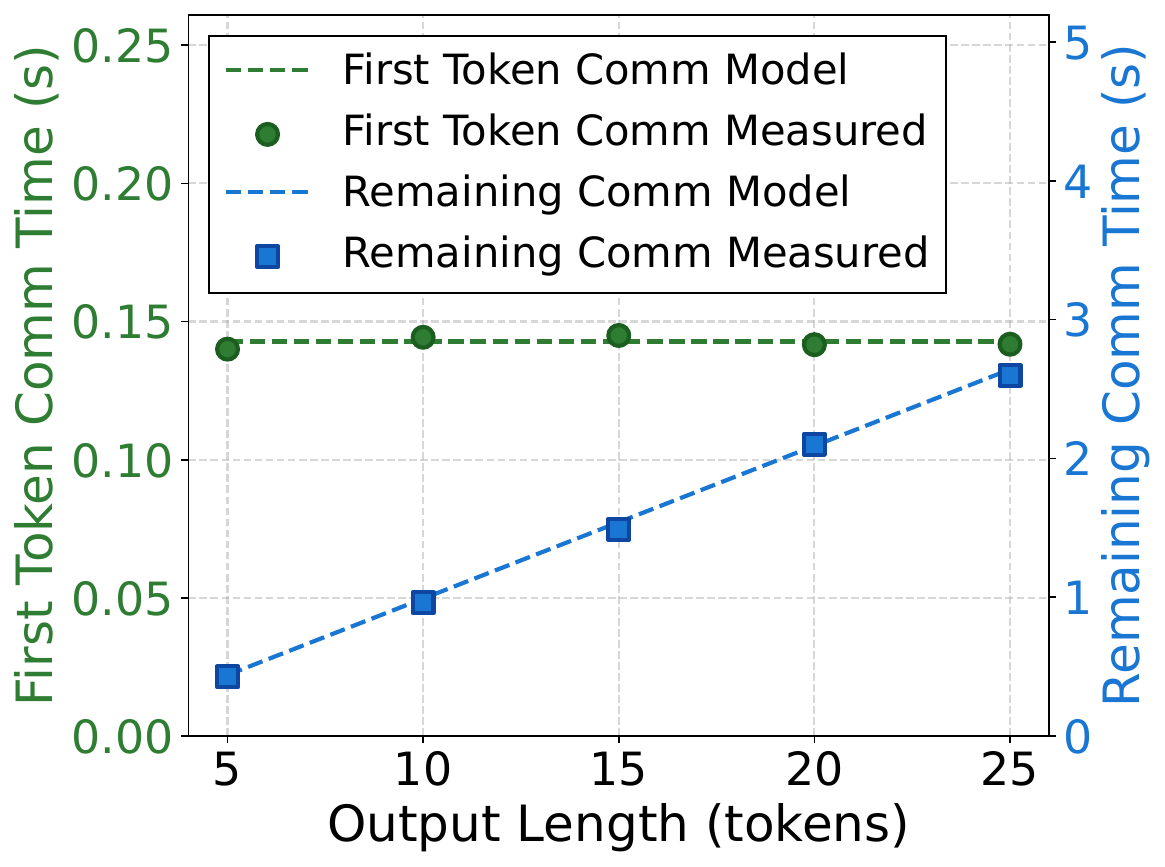}
    \vspace{-1.75em}
    \caption{(a) Comm time vs. output length}
    \label{fig:comm_varying_output}
\end{subfigure}
\begin{subfigure}[t]{0.4\linewidth}
    \centering
    \includegraphics[width=\linewidth]{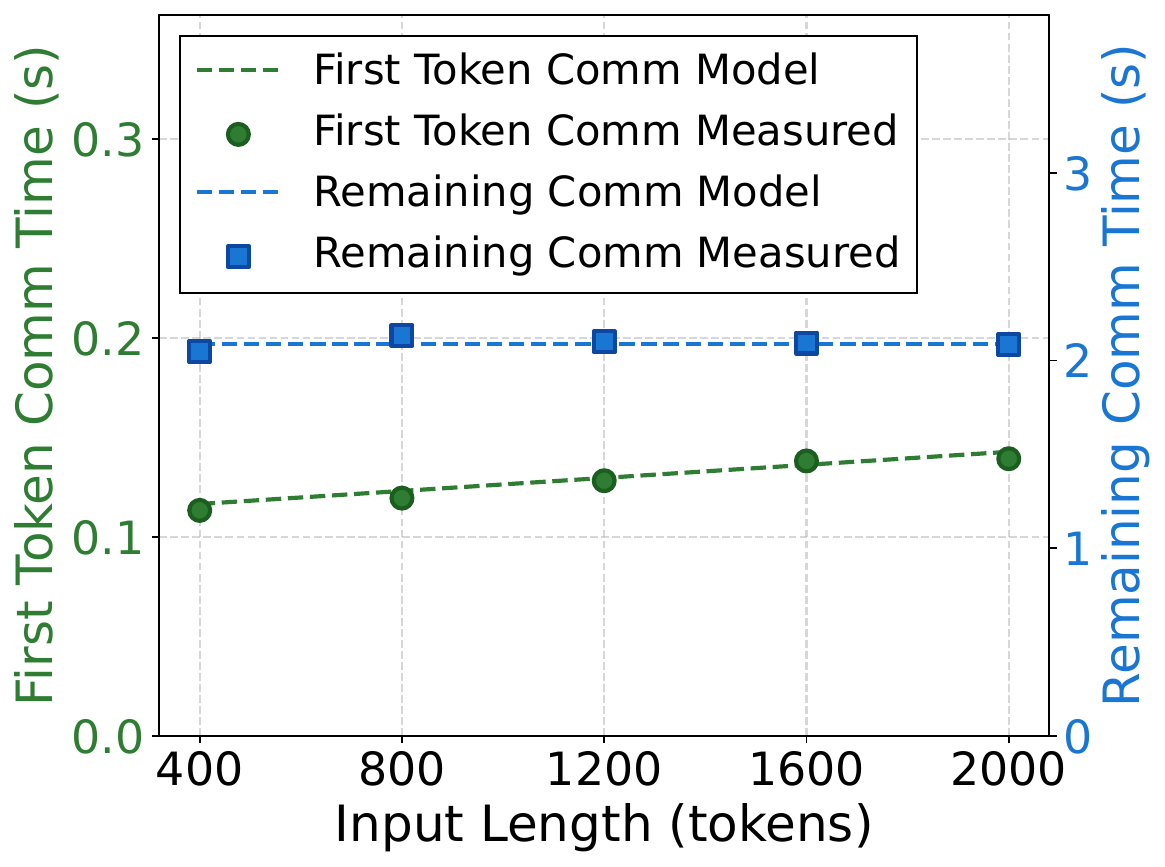}
    \vspace{-1.75em}
    \caption{(b) Comm time vs. input length}
    \label{fig:comm_varying_input}
\end{subfigure}
\vspace{-0.75em}
\caption{Communication time profiling (by default, input length is 2,000 and output length is 20; assume 100 ms RTT and 1 Gbps bandwidth). }
\label{fig:comm_time_profiling}
\end{figure}

\fi

\if\thisismainpaper 1
We first estimate the parameters in our system model (e.g., $\tau^p_j$, $\tau^c_j$) based on measurements from the PETALS deployment, as shown in Section~4.2.2 of \cite{Sun26arXiv}. 
\fi 

We then test the assumptions of Poisson arrivals and exponentially distributed service times used in our mean response time analysis (Section~\ref{subsubsec:Response Time Analysis}) against measurements based on the trace. The results in Fig.~\ref{fig:assumption_validation} indicate that the actual distributions notably deviate from our theoretical assumptions: (i) the inter-arrival times are more bursty than exponential distribution with the same mean (the std ratio between the actual and the exponential distributions is 13.15), while (ii) the service times are less bursty than exponential distribution with the same mean (the std ratios is 0.71--0.81).  
Nevertheless, our proposed solution configured under these simplifying assumptions remains highly effective as shown below. 

\begin{figure}[t]
\centering
\begin{subfigure}[t]{0.4\linewidth}
\centering
\includegraphics[width=\linewidth]{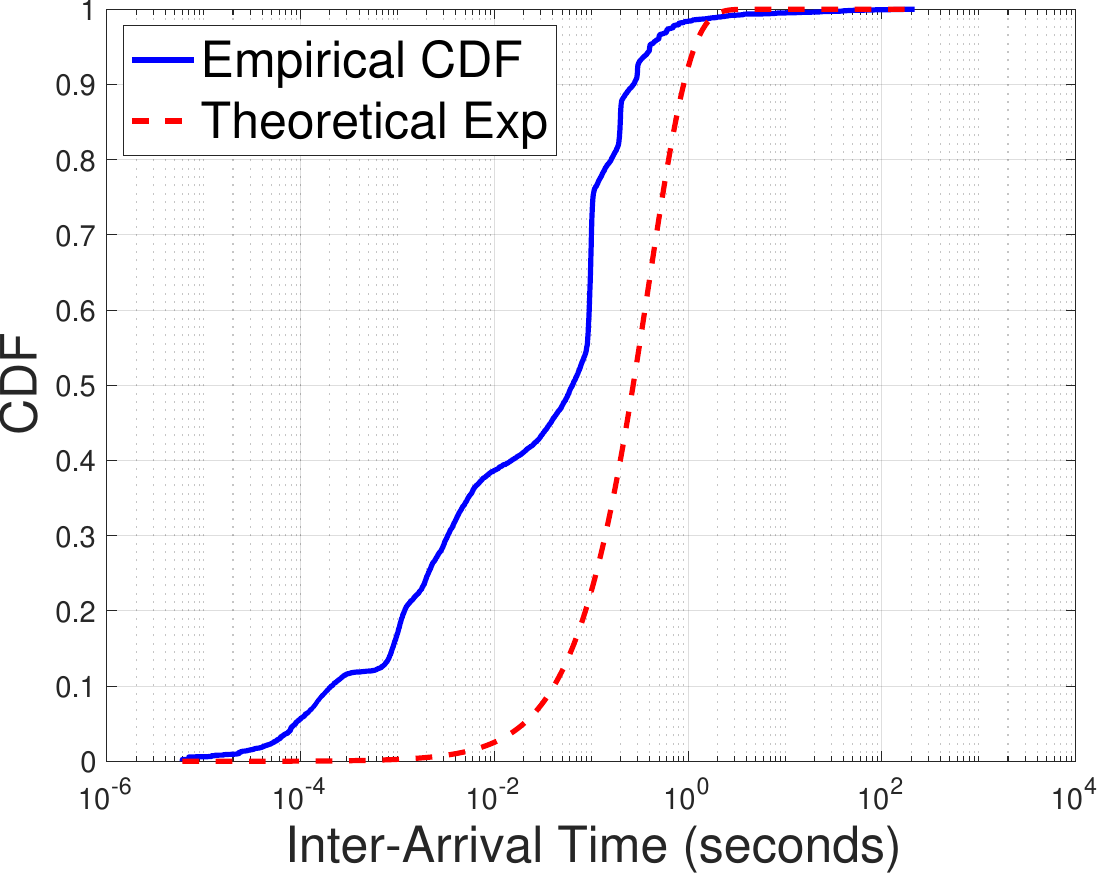}
\caption{(a) Inter-arrival times}
\label{fig:poisson_arrival_cdf}
\end{subfigure}
\begin{subfigure}[t]{0.4\linewidth}
\centering
\includegraphics[width=\linewidth]
{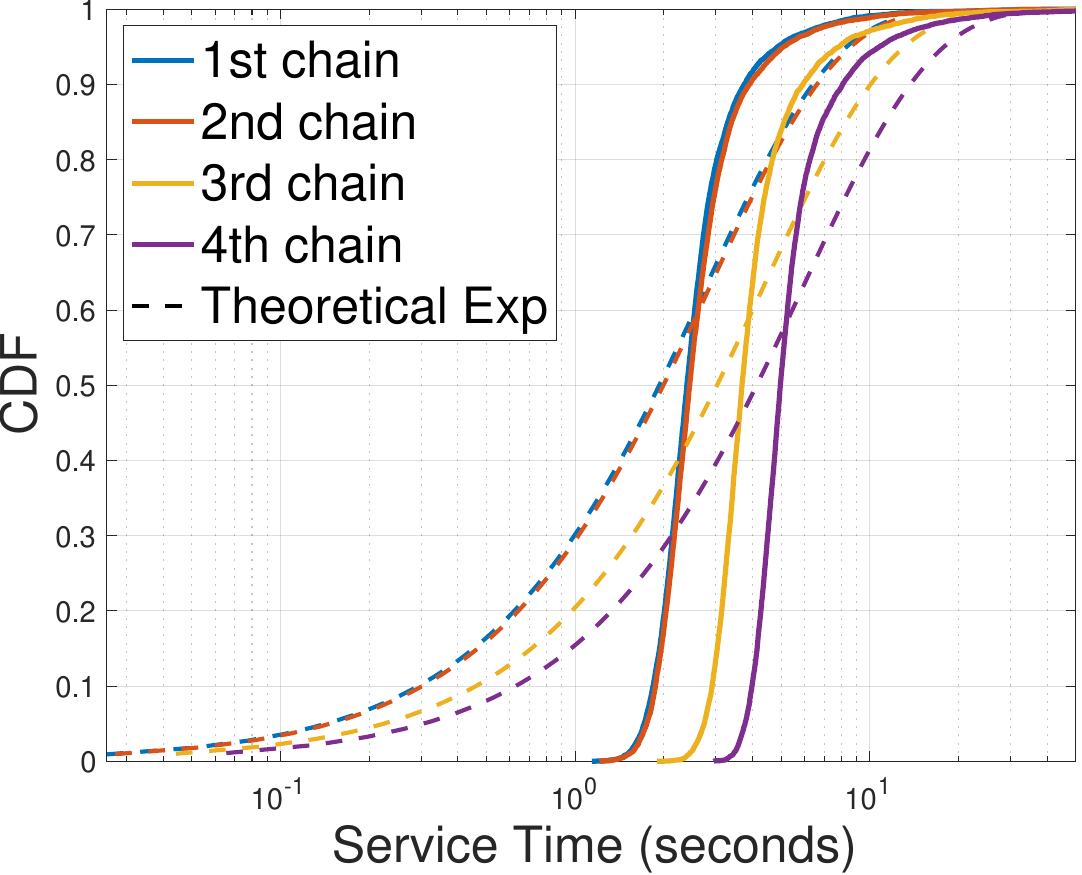}
\caption{(b) Service times}
\label{fig:service_time_cdf}
\end{subfigure}
\vspace{-0.75em}
\caption{Empirical CDF vs.\ exponential distribution CDF for the Azure LLM inference trace: (a) Inter-arrival times; 
(b) Service times on the fastest $K$ server chains (solid lines: empirical; dashed lines: exponential). 
} 
\label{fig:assumption_validation}
\end{figure}

\subsubsection{Performance Comparison}

\begin{table}[!t]\centering
\small
\begin{tabular}{lcccc}
\toprule
\textbf{Metric} & \textbf{PETALS~\cite{Borzunov23NeurIPS}} & \textbf{BPRR~\cite{Sun25Performance}} & \textbf{JFFC only} & \textbf{Proposed} \\
\midrule
\multicolumn{5}{l}{\textit{Response Time (seconds)}} \\
\quad Mean & $31.4$ 
& $19.8$ 
& $10.0$ 
& $\mathbf{7.3}$ \\ 
\quad Median & $27.8$ & $16.9$ & $8.5$ & $\mathbf{6.5}$ \\
\quad P95 & $68.5$ & $44.2$ & $22.1$ & $\mathbf{15.2}$ \\
\quad P99 & $89.3$ & $61.7$ & $29.6$ & $\mathbf{21.4}$ \\
\midrule
\multicolumn{5}{l}{\textit{Waiting Time (seconds)}} \\
\quad Mean & $24.2$ 
& $12.6$ 
& $1.5$ 
& $\mathbf{0.6}$ \\ 
\quad Median & $20.5$ & $9.7$ & $0.1$ & $\mathbf{0.1}$ \\
\quad P95 & $61.3$ & $37.1$ & $7.8$ & $\mathbf{4.8}$ \\
\quad Max & $142.6$ & $98.4$ & $25.1$ & $\mathbf{24.3}$ \\
\midrule
\multicolumn{5}{l}{\textit{Service Time (seconds)}} \\
\quad Mean & $7.2$ 
& $7.2$ 
& $8.5$ 
& $\mathbf{6.7}$ \\ 
\quad Median & $5.8$ & $5.9$ & $7.2$ & $\mathbf{5.7}$ \\
\quad P95 & $18.9$ & $18.7$ & $20.3$ & $\mathbf{18.2}$ \\
\quad Min / Max & $0.8 / 52.3$ & $0.9 / 51.8$ & $1.1 / 55.8$ & $\mathbf{0.8 / 49.1}$ \\
\midrule\midrule
\multicolumn{5}{l}{\textit{Improvement vs.\ PETALS}} \\
\quad Mean Response Time & -- & $36.9\%$ & $68.2\%$ & $\mathbf{76.8\%}$ \\
\quad Mean Waiting Time & -- & $47.9\%$ & $93.8\%$ & $\mathbf{97.5\%}$ \\
\quad P95 Response Time & -- & $35.5\%$ & $67.7\%$ & $\mathbf{77.8\%}$ \\
\bottomrule
\end{tabular}
\caption{Performance comparison for LLaMA-2-7B inference on 9 MIG instances ($3\times$ 3g.40gb + $6\times$ 2g.20gb) for 1000 requests generated from Azure LLM inference trace~\cite{Patel2024Splitwise}.}
\label{tab:performance_comparison}
\end{table}

\begin{table}[!t]\centering
\small
\begin{tabular}{lcccc}
\toprule
\textbf{Metric} & \textbf{PETALS~\cite{Borzunov23NeurIPS}} & \textbf{BPRR~\cite{Sun25Performance}} & \textbf{JFFC only} & \textbf{Proposed} \\
\midrule
\multicolumn{5}{l}{\textit{Response Time (seconds)}} \\
\quad Mean & $32.3$ & $20.3$ & $10.2$ & $\mathbf{7.4}$ \\
\quad Median & $28.1$ & $17.0$ & $8.5$ & $\mathbf{6.6}$ \\
\quad P95 & $71.9$ & $45.9$ & $22.9$ & $\mathbf{15.6}$ \\
\quad P99 & $94.6$ & $62.4$ & $29.1$ & $\mathbf{22.3}$ \\
\midrule
\multicolumn{5}{l}{\textit{Waiting Time (seconds)}} \\
\quad Mean & $24.9$ & $12.9$ & $1.6$ & $\mathbf{0.7}$ \\
\quad Median & $20.7$ & $9.8$ & $0.1$ & $\mathbf{0.1}$ \\
\quad P95 & $63.6$ & $38.0$ & $8.3$ & $\mathbf{5.1}$ \\
\quad Max & $151.2$ & $101.9$ & $26.7$ & $\mathbf{25.2}$ \\
\midrule
\multicolumn{5}{l}{\textit{Service Time (seconds)}} \\
\quad Mean & $7.4$ & $7.3$ & $8.6$ & $\mathbf{6.8}$ \\
\quad Median & $5.9$ & $6.0$ & $7.3$ & $\mathbf{5.8}$ \\
\quad P95 & $19.6$ & $19.2$ & $20.6$ & $\mathbf{18.5}$ \\
\quad Min / Max & $0.8 / 54.9$ & $0.9 / 53.4$ & $1.1 / 57.2$ & $\mathbf{0.8 / 50.3}$ \\
\midrule\midrule
\multicolumn{5}{l}{\textit{Improvement vs.\ PETALS}} \\
\quad Mean Response Time & -- & $37.2\%$ & $68.4\%$ & $\mathbf{77.1\%}$ \\
\quad Mean Waiting Time & -- & $48.2\%$ & $93.6\%$ & $\mathbf{97.2\%}$ \\
\quad P95 Response Time & -- & $36.2\%$ & $68.2\%$ & $\mathbf{78.3\%}$ \\
\bottomrule
\end{tabular}
\caption{Performance comparison under network jitter for LLaMA-2-7B inference on 9 MIG
instances ($3\times$ 3g.40gb + $6\times$ 2g.20gb) for 1000 requests generated from
Azure LLM inference trace~\cite{Patel2024Splitwise}. 
}
\label{tab:performance_comparison_jitter}
\end{table}

We compare the response times between the proposed solution and the benchmarks described in Section~\ref{subsubsec:Overall Comparison - simulation}. The results without jitter, given in Table~\ref{tab:performance_comparison}, show substantial performance difference for different resource allocation algorithms, where BPRR~\cite{Sun25Performance} reduces the mean response time by $36.9\%$ compared to the current algorithms used in PETALS~\cite{Borzunov23NeurIPS}, and the proposed solution increases such reduction to $76.8\%$ (a $63.1\%$ reduction compared to BPRR). A closer look at the response time components shows that most of the improvement comes from reducing the waiting times, while the service times remain similar. Note that although the model used in this experiment is small enough to fit in an individual server, the high arrival rate in the trace leads to a large amount of cache reservation ($c=35$ according to the lower bound in Theorem~\ref{thm:mean occupancy bounds}), which causes the model to be split across multiple servers\footnote{LLaMA-2-7B has 32 transformer blocks. Under $c=35$, each 3g.40gb can only hold 24 blocks and each 2g.20gb can only hold 12 blocks.}. 
To test the efficacy of our cache reservation strategy, we further evaluate a benchmark that places an entire model instance onto each server and uses JFFC for load balancing. The results (`JFFC only' in Table~\ref{tab:performance_comparison}) demonstrate an intermediate response time performance between the proposed solution and the benchmarks from \cite{Borzunov23NeurIPS,Sun25Performance}. Closer examination shows that compared to the proposed solution, placing the entire model onto each server leads to less cache space at inference time, which reduces the parallel job-processing capability of faster servers that forces more requests to be assigned to slower servers (which increases service times) or queued  (which increases waiting times). 
\if\thisismainpaper 0
Moreover, requiring a single server to process all the blocks also increases the cache requirement for each active job that leads to lower memory utilization. In our experiment, the KV cache requirement on a server processing all the blocks is $\sim 2$~GiB per active inference session. In this case, a server at its concurrency limit can still have free-but-unusable memory of up to $2$~GiB, which amounts to $10\%$ of the HBM on 2g.20gb and $5\%$ on 3g.40gb. 

Finally, since our server chain composition only depends on the mean communication time $\tau^c_j$ as in \eqref{eq:mean service time model} and is thus agnostic to short-term delay fluctuations, we repeat the experiment under the jitter configuration in Section~\ref{subsubsec:Experiment Settings} to test whether this abstraction is a limitation in practice. The results in Table~\ref{tab:performance_comparison_jitter} show that moderate jitter barely affects the mean performance, as the mean response time increases by less than $2\%$ for every solution. The tail performance is affected more, but the proposed solution degrades less than PETALS~\cite{Borzunov23NeurIPS}. 
Its improvement over PETALS~\cite{Borzunov23NeurIPS} is hence preserved under jitter, at $77.1\%$ in mean response time and $78.3\%$ in 95th-percentile response time.
\fi

\section{Conclusion}\label{sec:Conclusion}

We have addressed a resource allocation problem in serving chain-structured jobs with large memory footprints by composable server chains, inspired by pipeline-parallel processing of inference requests for large transformers. Based on the unique requirements of such workloads, 
we formulate and solve a set of novel optimizations aimed at optimally composing server chains to minimize the mean response time. Our experiments based on a real LLM serving system and real demand traces show that the proposed solution is highly effective in reducing the response time compared to state-of-the-art solutions. 
Meanwhile, fully optimizing large model serving involves many other considerations, such as time-varying demands and additional control knobs allowed by modern model serving systems. 
Our optimization-based approach has laid a solid foundation for addressing these problems in future research.  \looseness=0

\bibliographystyle{plain}
\bibliography{references}

@article{Sun26arXiv,
  author  = {Tingyang Sun and Ting He and I-Hong Hou},
  title   = {Serving Chain-structured Jobs with Large Memory Footprints with Application to Large Foundation Model Serving},
  journal = {CoRR},
  year    = {2026},
  doi     = {10.48550/arXiv.2604.14993},
  url     = {https://doi.org/10.48550/arXiv.2604.14993}
}

@article{Sun25Performance,
title = {Optimizing resource allocation for geographically-distributed inference by large language models},
journal = {Performance Evaluation},
volume = {170},
pages = {102527},
year = {2025},
author = {Tingyang Sun and Ting He and Bo Ji and Parimal Parag}
}

@article{Bhambay22PE,
title = {Asymptotic optimality of speed-aware JSQ for heterogeneous service systems},
journal = {Performance Evaluation},
volume = {157-158},
pages = {102320},
year = {2022},
author = {Sanidhay Bhambay and Arpan Mukhopadhyay}
}

@article{Weng20MACS,
author = {Weng, Wentao and Zhou, Xingyu and Srikant, R.},
title = {Optimal Load Balancing with Locality Constraints},
year = {2020},
volume = {4},
number = {3},
journal = {Proc. ACM Meas. Anal. Comput. Syst.},
numpages = {37}
}

@article{Mitzenmacher25SS,
    author = {Michael Mitzenmacher and Rana Shahout},
    title = {Queueing, Predictions, and Large Language Models: Challenges and Open Problems},
    journal = {Stochastic Systems},
    year = {2025},
 volume = {15},
 number = {3}, 
 pages = {195-219}
}

@Article{Ainbinder24AS,
AUTHOR = {Ainbinder, Inessa and Temnikov, Evgeni and Allalouf, Miriam},
TITLE = {A Study Comparing Waiting Times in Global and Local Queuing Systems with Heterogeneous Workers},
JOURNAL = {Applied Sciences},
VOLUME = {14},
YEAR = {2024},
NUMBER = {9},
ARTICLE-NUMBER = {3799}
}

@article{Winston77JAP,
 author = {Wayne Winston},
 journal = {Journal of Applied Probability},
 number = {1},
 pages = {181--189},
 title = {Optimality of the Shortest Line Discipline},
 volume = {14},
 year = {1977}
}

@article{Lu11PE,
title = {Join-Idle-Queue: A novel load balancing algorithm for dynamically scalable web services},
journal = {Performance Evaluation},
volume = {68},
number = {11},
pages = {1056-1071},
year = {2011},
author = {Yi Lu and Qiaomin Xie and Gabriel Kliot and Alan Geller and James R. Larus and Albert Greenberg}
}

@article{Wang18TON,
    author = {Chunpu Wang and Chen Feng and Julian Cheng},
    title = {Distributed Join-the-Idle-Queue for Low Latency Cloud Services},
    journal = {IEEE/ACM Transactions on Networking},
    volume = {26},
    number = {5},
    year = {2018}, 
    pages = {2309–2319}
}

@inproceedings{Mitzenmacher96FOCS,
author = {Mitzenmacher, Michael},
title = {Load balancing and density dependent jump Markov processes},
year = {1996},
booktitle = {Proceedings of the 37th Annual Symposium on Foundations of Computer Science},
pages = {213}
}

@article{Gromicho12COR,
title = {Solving the job-shop scheduling problem optimally by dynamic programming},
journal = {Computers \& Operations Research},
volume = {39},
number = {12},
pages = {2968-2977},
year = {2012},
author = {Joaquim A.S. Gromicho and Jelke J. {van Hoorn} and Francisco Saldanha-da-Gama and Gerrit T. Timmer}
}

@inproceedings{Huang19NeurIPS,
 author = {Huang, Yanping and Cheng, Youlong and Bapna, Ankur and others},
 booktitle = {Advances in Neural Information Processing Systems},
 title = {GPipe: Efficient Training of Giant Neural Networks using Pipeline Parallelism},
 url = {https://proceedings.neurips.cc/paper_files/paper/2019/file/093f65e080a295f8076b1c5722a46aa2-Paper.pdf},
 volume = {32},
 year = {2019}
}

@inproceedings{Narayanan19SOSP,
author = {Narayanan, Deepak and Harlap, Aaron and Phanishayee, Amar and others},
title = {PipeDream: generalized pipeline parallelism for DNN training},
year = {2019},
booktitle = {Proceedings of the 27th ACM Symposium on Operating Systems Principles},
pages = {1–15}
}

@inproceedings{Yang21MLsys,
 author = {Yang, Bowen and Zhang, Jian and Li, Jonathan  and others},
 booktitle = {Proceedings of Machine Learning and Systems},
 pages = {269--296},
 title = {PipeMare: Asynchronous Pipeline Parallel DNN Training},
 year = {2021}
}

@article{Krizhevsky17CommACM,
author = {Krizhevsky, Alex and Sutskever, Ilya and Hinton, Geoffrey E.},
title = {{ImageNet} classification with deep convolutional neural networks},
year = {2017},
volume = {60},
number = {6},
journal = {Communications of the ACM},
month = {May},
pages = {84–90}
}

@article{Ben-Nun19ACMCompSurv,
author = {Ben-Nun, Tal and Hoefler, Torsten},
title = {Demystifying Parallel and Distributed Deep Learning: An In-depth Concurrency Analysis},
year = {2019},
volume = {52},
number = {4},
journal = {ACM Comput. Surv.},
month = {Aug},
articleno = {65},
numpages = {43}
}

@article{Tang20arXiv,
  author       = {Zhenheng Tang and
                  Shaohuai Shi and
                  Xiaowen Chu and
                  Wei Wang and
                  Bo Li},
  title        = {Communication-Efficient Distributed Deep Learning: {A} Comprehensive
                  Survey},
  journal      = {CoRR},
  volume       = {abs/2003.06307},
  year         = {2020} 
}

@ARTICLE{Medhat17CommMag,
  author={Medhat, Ahmed M. and Taleb, Tarik and Elmangoush, Asma and others},
  journal={IEEE Communications Magazine}, 
  title={Service Function Chaining in Next Generation Networks: State of the Art and Research Challenges}, 
  year={2017},
  volume={55},
  number={2},
  pages={216-223}
}

@ARTICLE{Jang17JSAC,
  author={Jang, Insun and Suh, Dongeun and Pack, Sangheon and Dán, György},
  journal={IEEE Journal on Selected Areas in Communications}, 
  title={Joint Optimization of Service Function Placement and Flow Distribution for Service Function Chaining}, 
  year={2017},
  volume={35},
  number={11},
  pages={2532-2541}
}

@INPROCEEDINGS{Zhang17ICDCS,
  author={Zhang, Qixia and Xiao, Yikai and Liu, Fangming and Lui, John C.S. and Guo, Jian and Wang, Tao},
  booktitle={IEEE 37th International Conference on Distributed Computing Systems (ICDCS)}, 
  title={Joint Optimization of Chain Placement and Request Scheduling for Network Function Virtualization}, 
  year={2017},
  volume={},
  number={},
  pages={731-741}
}

@INPROCEEDINGS{Guo18INFOCOM,
  author={Guo, Linqi and Pang, John and Walid, Anwar},
  booktitle={IEEE INFOCOM 2018 - IEEE Conference on Computer Communications}, 
  title={Joint Placement and Routing of Network Function Chains in Data Centers}, 
  year={2018},
  volume={},
  number={},
  pages={612-620}
}

@INPROCEEDINGS{Jalalitabar19NFV-SDN,
  author={Jalalitabar, Maryam and Wang, Yang and Cao, Xiaojun},
  booktitle={IEEE Conference on Network Function Virtualization and Software Defined Networks (NFV-SDN)}, 
  title={Branching-Aware Service Function Placement and Routing in Network Function Virtualization}, 
  year={2019},
  volume={},
  number={},
  pages={1-6}
}

@INPROCEEDINGS{Ma17INFOCOM,
  author={Ma, Wenrui and Sandoval, Oscar and Beltran, Jonathan and Pan, Deng and Pissinou, Niki},
  booktitle={IEEE Conference on Computer Communications (INFOCOM)}, 
  title={Traffic aware placement of interdependent NFV middleboxes}, 
  year={2017},
  volume={},
  number={},
  pages={1-9}
}

@article{Chen19arXiv,
  author       = {Weihan Chen and
                  Xia Yin and
                  Zhiliang Wang and
                  Xingang Shi},
  title        = {Placement and Routing Optimization Problem for Service Function Chain:
                  State of Art and Future Opportunities},
  journal      = {CoRR},
  volume       = {abs/1910.02613},
  year         = {2019}  
}

@inproceedings{Shang19ICPP,
author = {Shang, Xiaojun and Liu, Zhenhua and Yang, Yuanyuan},
title = {Network Congestion-aware Online Service Function Chain Placement and Load Balancing},
year = {2019},
booktitle = {Proceedings of the 48th International Conference on Parallel Processing},
articleno = {46},
numpages = {10}
}

@misc{BigScience23BLOOM,
      title={{BLOOM}: A {176B}-Parameter Open-Access Multilingual Language Model}, 
      author={Teven Le Scao and Angela Fan and Christopher Akiki and others},
      year={2023},
      eprint={2211.05100},
      archivePrefix={arXiv},
      primaryClass={cs.CL},
      url={https://arxiv.org/abs/2211.05100}, 
}

@inproceedings{Borzunov23NeurIPS,
 author = {Borzunov, Alexander and Ryabinin, Max and Chumachenko, Artem and others},
 booktitle = {Advances in Neural Information Processing Systems},
 pages = {12312--12331},
 title = {Distributed Inference and Fine-tuning of Large Language Models Over The Internet},
 volume = {36},
 year = {2023}
}

@article{Hayes02AS,
 author = {Brian Hayes},
 journal = {American Scientist},
 number = {2},
 pages = {113--117},
 title = {Computing Science: The Easiest Hard Problem},
 volume = {90},
 year = {2002}
}

@Inbook{Kellerer2004BookChapter,
author="Kellerer, Hans
and Pferschy, Ulrich
and Pisinger, David",
title="Multidimensional Knapsack Problems",
bookTitle="Knapsack Problems",
year="2004",
publisher="Springer Berlin Heidelberg",
pages="235--283"
}

@misc{vLLM_Ray,
  author       = {{A. Sprenger}},
  title        = {Ray vLLM Inference},
  howpublished = {\url{https://github.com/asprenger/ray_vllm_inference}},
  note         = {Accessed: 2025-09-01}
}

@misc{Nvidia_Dynamo,
  author       = {{NVIDIA}},
  title        = {A Datacenter Scale Distributed Inference Serving Framework},
  howpublished = {\texttt{https://github.com/ai-dynamo/dynamo}},
  note         = {Accessed: 2025-09-01}
}

@misc{Amazon_EKS,
  author       = {{Amazon Web Services}},
  title        = {Amazon Elastic Kubernetes Service},
  howpublished = {\texttt{https://aws.amazon.com/eks/}},
  note         = {Accessed: 2025-09-01}
}

@misc{vLLM_distributed,
  author       = {{vLLM Project}},
  title        = {vLLM: Distributed Inference and Serving},
  howpublished = {\texttt{https://docs.vllm.ai/en/stable/index.html}},
  note         = {Accessed: 2025-09-01}
}

@inproceedings{Patel2024Splitwise,
  title = {Splitwise: Efficient Generative LLM Inference Using Phase Splitting},
  author = {Pratyush Patel and Esha Choukse and Chaojie Zhang and Aashaka Shah and Íñigo Goiri and Saeed Maleki and Ricardo Bianchini},
  booktitle = {Proceedings of the 51st Annual International Symposium on Computer Architecture (ISCA)},
  year = {2024},
  pages = {123--136},
  publisher = {ACM/IEEE},
  address = {Toronto, Canada},
  doi = {10.1145/3456789.3456790}
}

@MISC{NvidiaMIG,
  author = {NVIDIA Corporation},
  title = {NVIDIA Multi-Instance GPU User Guide},
  howpublished = {\url{https://docs.nvidia.com/datacenter/tesla/mig-user-guide/index.html}},
  year = {2025}
}

@article{schubert2019network,
  title={Network emulation using Linux network namespaces},
  author={Schubert, Daniel and Jaeger, Benedikt and Helm, Max},
  journal={Network},
  volume={57},
  year={2019}
}

@misc{ripe_atlas,
  title = {{RIPE Atlas}: A Global Internet Measurement Network},
  author = {{RIPE NCC}},
  howpublished = {\url{https://atlas.ripe.net/}},
  year = {2024}
}

\appendix
\section{Appendix}\label{appendix:A}

\subsection{Notation}\label{appendix:Notations}

Table~\ref{tab:notations} summarizes the main notation used in this paper, where the first half are input parameters and the second half are decision variables (including both independent and dependent variables). 

\begin{table}
\small
    \centering
    \begin{tabular}{c|l}
      Notation   &   Description \\
         \hline
$J$ &  number of servers  \\
$\mathcal{J}$, $\mathcal{J}_+$ & set/extended set of servers \\
$M_j$ & memory size of server $j$ \\
$\tau^c_j$ & mean communication time of server $j$ \\
$\tau^p_j$ & mean per-block computation time of server $j$ \\
$L$ & number of blocks in the service \\
$s_m$, $s_c$ & size of each block/cache slot \\
$\lambda$ & request arrival rate \\
$\overline{\rho}$ & target maximum system load \\
\hline
$a_j$, $m_j$ & first block/number of blocks at server $j$ \\
$m_{ij}$ & number of blocks processed at server $j$ after server $i$ \\
$\widetilde{M}_j$ & number of cache slots at server $j$ after block placement \\
$\mathcal{E}_{\bm{a},\bm{m}}$ & feasible neighboring servers under block placement $(\bm{a},\bm{m})$ \\
$\mathcal{K}_{\bm{a},\bm{m}}$ & feasible server chains under block placement $(\bm{a},\bm{m})$ \\
$c_k$, $\mu_k$ & capacity/service rate of chain $k$ \\
$c$ & required capacity of each server (during block placement)
    \end{tabular}
    \vspace{-.5em}
    \caption{Main notation.      
    }
    \label{tab:notations}
    \vspace{-.0em}
\end{table}

\subsection{Supporting Proofs}\label{appendix:Proofs}

\begin{proof}[Proof of Theorem~\ref{thm:NP-hardness of BPCA}]
We prove the conclusion by reducing the {Multidimensional Knapsack Problem (MKP)} to an instance of our problem. According to \cite{Kellerer2004BookChapter}, the MKP has the following form:
\begin{subequations}\label{eq:MKP}
\begin{align}
\max_{\bm{c}}\quad & \sum_{k=1}^K \mu_k c_k \\
\mbox{s.t.}\quad & \sum_{k=1}^K m_{jk} c_k \leq D_j,~~~\forall j\in [d],\\
& c_k\in \{0,1\},~~~\forall k\in [K],
\end{align}
\end{subequations}
which tries to select from a given set of $K$ $d$-dimensional items, each with value $\mu_k$ and size $m_{jk}$ in dimension $j$, to maximize the total value while fitting into a knapsack with capacity $D_j$ in dimension $j$, where $\mu_k$ and $D_j$ are positive integers and $m_{jk}$ is a nonnegative integer. 

We construct a corresponding instance of our problem: each item $k$ corresponds to a feasible server chain with service rate $\mu_k$; each dimension $j$ corresponds to a server with $\widetilde{M}_j = D_j$ cache slots after block placement; each server $j\in [d]$ hosts $m_j:= \max_{k\in [K]} m_{jk}$ blocks $\{1+\sum_{j'=1}^{j-1}m_{j'},\ldots,\sum_{j'=1}^j m_{j'}\}$; the total number of blocks is $L:= 1+  \sum_{j\in [d]}m_j$. In addition, for each $k\in [K]$ and each $j\in [d]$ with $m_{jk}<m_j$, add a server $v_{jk}$ before $j$ onto chain $k$ that hosts the same first $m_j-m_{jk}$ blocks as $j$ with $m_j-m_{jk}$ cache slots; for each chain $k\in [K]$, add a server at the end that hosts the last block with $1$ cache slot. 
By the above construction, we have $K$ server chains that can each run up to one job at a time, and enabling chain $k$ to run a job requires $m_{jk}$ cache slots at server $j\in [d]$. 
Thus, the decision version of MKP ``is there a feasible solution to \eqref{eq:MKP} with objective value $\geq \mu$?'' is equivalent to the feasibility test of the constructed subproblem of \eqref{eq:BPCA - conceptual} with the $K$ constructed server chains and $\lambda/\overline{\rho}:=\mu$. As the decision version of MKP is NP-complete~\cite{Kellerer2004BookChapter}, the subproblem of \eqref{eq:BPCA - conceptual} is NP-hard.  
\end{proof}

\begin{proof}[Proof of Lemma~\ref{lem:service rate bound}]
Since $\sum_{j\in \mathcal{J}_k}m_j(c)\geq L$, we can place the entire set of blocks onto each subset of servers $\mathcal{J}_k$ ($\forall k\in \mathcal{K}$). Let $k$ denote the chain corresponding to $\mathcal{J}_k$. Each chain $k$ has a mean service time no greater than $\sum_{j\in \mathcal{J}_k}t_j(c)$, as the mean time spent at each $j\in\mathcal{J}_k$ is upper-bounded by $t_j(c)$. The service rate of chain $k$ is thus no less than $(\sum_{j\in \mathcal{J}_k}t_j(c))^{-1}$. Moreover, according to \eqref{eq:m_j(c)}, this block placement leaves enough residual memory at each server to process $c$ jobs concurrently. This completes the proof.
\end{proof}

\begin{proof}[Proof of Lemma~\ref{lem:NP-hardness of server grouping}]
We prove the claim via a reduction from the \emph{partition problem}~\cite{Hayes02AS}. Given a set of positive integers $\mathcal{X}:=\{x_1,\ldots,x_N\}$ where $N$ is an even number, the partition problem determines if there exists a partition of $\mathcal{X}$ into two subsets $\mathcal{X}_1, \mathcal{X}_2$ with equal sum. 
We construct a corresponding instance of our problem \eqref{eq:BP} by: (i) constructing a server $j$ for each number $x_j$ with $m_j(c)=t_j(c)=x_j$, and (ii) setting $L:= (\sum_{j=1}^N x_j)/2$ and $\lambda/(\overline{\rho}c):= 2/L$. 
If there is no feasible solution to the partition problem, then the constructed problem has only one server group $\mathcal{J}'\subseteq [N]$ that satisfies \eqref{BP:feasibility}, with a scaled total service rate of 
\begin{align}
{1\over \sum_{j\in \mathcal{J}'} x_j} \leq {1\over L},
\end{align} 
which cannot satisfy the constraint \eqref{BP:service rate}. 
If there is a feasible solution $(\mathcal{X}_1, \mathcal{X}_2)$ to the partition problem, then $(\mathcal{J}_1,\mathcal{J}_2)$ with $\mathcal{J}_k:=\{j:\: x_j\in \mathcal{X}_k\}$ is a feasible solution to \eqref{eq:BP}. Thus, testing the feasibility of the constructed instance of \eqref{eq:BP} will solve the partition problem. Since the partition problem is NP-complete~\cite{Hayes02AS}, \eqref{eq:BP} is NP-hard. 
\end{proof}

\begin{proof}[Proof of Theorem~\ref{thm:optimality of GBP-CR}]
Since $M_j\equiv M$, $m_j(c)\equiv m(c)$, and constraint \eqref{BP:feasibility} reduces to $|\mathcal{J}_k| \geq \lceil L/m(c) \rceil =: n(c)$. Let $(\mathcal{J}^*_k)_{k\in \mathcal{K}^*}$ be an optimal solution to \eqref{eq:BP}. Without loss of generality, we assume that $T_{k_1} := \sum_{j\in \mathcal{J}^*_{k_1}}t_j(c) \leq \sum_{j\in \mathcal{J}^*_{k_2}}t_j(c) =: T_{k_2}$ for any $k_1<k_2$ (interpreted as ``chain $k_1$ is constructed before chain $k_2$''), and $|\mathcal{J}^*_k|\equiv n(c)$ ($\forall k\in \mathcal{K}^*$). Suppose that $\exists k_1<k_2$ and $j_1\in \mathcal{J}^*_{k_1},\: j_2\in \mathcal{J}^*_{k_2}$ such that $t_{j_1}(c)> t_{j_2}(c)$. 
Then
\begin{align}
&{1\over T_{k_1} -t_{j_1}(c)+t_{j_2}(t)} + {1\over T_{k_2} -t_{j_2}(c)+t_{j_1}(c)}  - {1\over T_{k_1}} 
- {1\over T_{k_2}} \nonumber\\
=& {t_{j_1}(c)-t_{j_2}(c)\over (T_{k_1}-t_{j_1}(c)+t_{j_2}(c))T_{k_1}} - {t_{j_1}(c)-t_{j_2}(c)\over (T_{k_2} +t_{j_1}(c)-t_{j_2}(c))T_{k_2}} > 0,
\end{align}
i.e., swapping $j_1$ and $j_2$ can only increase the scaled total service rate and thus the resulting solution remains optimal. Thus, there must exist an optimal solution to \eqref{eq:BP} in which $t_{j_1}(c)\leq t_{j_2}(c)$ for all $k_1<k_2$ and $j_1\in \mathcal{J}^*_{k_1},\: j_2\in \mathcal{J}^*_{k_2}$. Moreover, it must be optimal to use the fastest $|\mathcal{K}^*|$ chains. Therefore, an optimal solution to \eqref{eq:BP} is to group servers into server chains in the increasing order of $t_j(c)$ until the scaled total service rate satisfies the requirement, which is exactly what is done in GBP-CR since $\widetilde{t}_j(c)\propto t_j(c)$ in this case. 
\end{proof}

\begin{proof}[Proof of Theorem~\ref{thm:GCA conditional optimality}]
First, suppose that when JFFS assigns a job at time $t$, the fastest chain in $\mathcal{K}$ with available capacity according to $\bm{c}$ is $k_l$. Then all the faster chains $k_1,\ldots,k_{l-1}$ in $\mathcal{K}$ must be fully occupied at this time, i.e., running $c_{k_1},\ldots,c_{k_{l-1}}$ jobs respectively. Thus, the set of feasible links $\mathcal{E}(t):=\{(i,j)\in \mathcal{E}_{\bm{a},\bm{m}}:\: j\mbox{ has at least }m_{ij}\mbox{ available cache slots at time }t\} \subseteq \mathcal{E}^{(l-1)}$. This implies that the length of the shortest $j_0\to j_{J+1}$ path in $\mathcal{G}(t)=(\mathcal{J}_+, \mathcal{E}(t))$ is no smaller than the length of the shortest $j_0\to j_{J+1}$ path in $\mathcal{G}^{(l-1)}=(\mathcal{J}_+, \mathcal{E}^{(l-1)})$, i.e., chain $k_l$ achieves the minimum mean service time among all the server chains with available capacity at time $t$, including those outside $\mathcal{K}$. Therefore, it suffices for JFFS to always assign jobs to chains in $\mathcal{K}$. 

Moreover, when $k_l\in \mathcal{K}$ is selected by JFFS for processing a job at time $t$, every faster chain $k_{l'}$ with $l'<l$ must be running $c_{k_{l'}}$ jobs, and thus the number of available cache slots $M_j(t)$ at this time on each server $j$ must be upper-bounded by $M^{(l-1)}_j$. This implies that the maximum number of jobs chain $k_l$ can run at time $t$ is no more than $c_{k_l}$. Thus, JFFS assigns no more than $c_k$ jobs on each $k\in \mathcal{K}$ at any point in time. 
\end{proof}

\begin{proof}[Proof of Lemma~\ref{lem:ergodicity}]
First, it is easy to see that any state $(n,\bm{c}_{1:K})$ in $\mathcal{Z}_2$ communicates with the state $(0,\bm{c}_{1:K})$ through arrivals and departures. Meanwhile, any state $(0,\bm{z}_{1:K})$ in $\mathcal{Z}_1$ can reach $(0,\bm{c}_{1:K})$ through arrivals, and can also be reached from $(0,\bm{c}_{1:K})$ through departures. Thus, $(\bm{Z}(t))_{t\geq 0}$ is irreducible. 

Next, we use the Foster–Lyapunov Criterion to prove positive recurrence. Define a Lyapunov function $V(\bm{z}):= z_0$ and a finite set $\mathcal{C}:= \mathcal{Z}_1$. Then for any $\bm{z}\not\in \mathcal{C}$, the Lyapunov drift satisfies
\begin{align}
\Delta V(\bm{z}) := \sum_{\bm{z}'}q(\bm{z},\bm{z}')(V(\bm{z}')-V(\bm{z})) = \lambda - \sum_{l=1}^K c_l\mu_l < 0.
\end{align}
For any $\bm{z}\in \mathcal{C}$, the Lyapunov drift satisfies
$\Delta V(\bm{z}) \leq \lambda$.  
Thus, by the Foster–Lyapunov Criterion, $(\bm{Z}(t))_{t\geq 0}$ is positive recurrent, which completes the proof. 
\end{proof}

\begin{proof}[Proof of Theorem~\ref{thm:mean occupancy bounds}]
It is easy to see that both constructed birth-death processes have unique steady-state distributions if $\lambda<\nu$. Let $\underline{\bm{\phi}}$ denote the steady-state distribution of $(\underline{\Phi}(t))_{t\geq 0}$ and $\overline{\bm{\phi}}$ denote the steady-state distribution of $(\overline{\Phi}(t))_{t\geq 0}$. We can verify that $\underline{\bm{\phi}}$ follows \eqref{eq:phi_n}. Specifically, by flow balance equations, we have
\begin{align}\label{eq:phi_n general}
\underline{\phi}_n = \left\{\begin{array}{ll} 
\underline{\phi}_0 \prod_{i=1}^n{\lambda\over \overline{\nu}_i} & \mbox{if }n\leq C,\\
\underline{\phi}_0 \left(\prod_{i=1}^C{\lambda\over \overline{\nu}_i}\right) \rho^{n-C} & \mbox{if }n>C.
\end{array}\right.
\end{align}
Then from $\sum_{n=0}^\infty \underline{\phi}_n = 1$, we can obtain
\begin{align}\label{eq:phi_0}
\underline{\phi}_0 = \left(1+\sum_{l=1}^{C-1}{\lambda^l\over \prod_{i=1}^l \overline{\nu}_i}+{\lambda^C \nu\over (\prod_{i=1}^C \overline{\nu}_i)(\nu-\lambda)} \right)^{-1}.
\end{align}
The definition of $\overline{\bm{\phi}}$ can be verified similarly. It is then easy to see that the steady-state mean $\mbbE_{\underline{\bm{\phi}}}[\underline{\Phi}]$ is given by the righthand side of \eqref{eq:E[Z] lower} and the steady-state mean $\mbbE_{\overline{\bm{\phi}}}[\overline{\Phi}]$ is given by \eqref{eq:E[Z] upper}. \looseness=0

We then show that the steady-state mean of such a birth-death process is decreasing in each of the death rates. Consider the birth-death process $(\Phi(t))_{t\geq 0}$ with a constant birth rate $\lambda$ and a state-dependent death rate $\nu_n$ defined in \eqref{eq:nu_n}. Its steady-state distribution $\bm{\phi}$ follows the same formula as \eqref{eq:phi_n general}--\eqref{eq:phi_0}, except that $\overline{\nu}_i$ is replaced by $\nu_i$. Define $b_n:=\phi_n/\phi_0$ ($\forall n\in \mbbN$). For any fixed $i\in\mbbN$, define $B_{i-}:=\sum_{n=0}^{i-1}b_n$ and $B_{i+}:=\sum_{n=i}^\infty b_n$. It is easy to see that $\partial b_n/\partial \nu_i = 0$ if $n<i$ and $-b_n/\nu_i$ if $n\geq i$. Thus, 
\begin{align}
{\partial \phi_n\over \partial \nu_i} = \left\{
\begin{array}{ll}
{\phi_0^2\over \nu_i} B_{i+} b_n & \mbox{if }n<i,\\
-{\phi_0^2\over \nu_i}B_{i-} b_n & \mbox{if }n\geq i.
\end{array}\right.
\end{align}
Plugging these into the formula for ${\partial\over \partial \nu_i}\mbbE_{\bm{\phi}}[\Phi]$ gives
\begin{align}
{\partial\over \partial \nu_i}\mbbE_{\bm{\phi}}[\Phi] = \sum_{n=0}^\infty n{\partial \phi_n\over \partial \nu_i} = {\phi_0^2\over \nu_i}B_{i+}\sum_{n=0}^{i-1}n b_n - {\phi_0^2\over \nu_i}B_{i-}\sum_{n=i}^\infty n b_n.
\end{align}
Since $\sum_{n=0}^{i-1}n b_n\leq (i-1)B_{i-}$ and $\sum_{n=i}^\infty n b_n \geq i B_{i+}$, we have
\begin{align}
{\partial\over \partial \nu_i}\mbbE_{\bm{\phi}}[\Phi]&\leq {\phi_0^2\over \nu_i}B_{i+}B_{i-}(i-1)-{\phi_0^2\over \nu_i}B_{i-}B_{i+}i \nonumber\\
&= -{\phi_0^2\over \nu_i}B_{i-}B_{i+} \leq 0,
\end{align}
which proves that $\mbbE_{\bm{\phi}}[\Phi]$ is decreasing in $\nu_i$ for any $i\in \mbbN$. 

The proof completes by noting that $\mbbE_{\bm{\pi}}\left[\sum_{l=0}^K Z_l\right] = \mbbE_{\bm{\phi}}[\Phi]$ by \eqref{eq:equivalence of occupancy}, and thus increasing the death rates to their upper bounds in \eqref{eq:nu_n upper} will lead to a lower bound on $\mbbE_{\bm{\pi}}\left[\sum_{l=0}^K Z_l\right]$ as in \eqref{eq:E[Z] lower} while decreasing the death rates to their lower bounds in \eqref{eq:nu_n lower} will lead to an upper bound on $\mbbE_{\bm{\pi}}\left[\sum_{l=0}^K Z_l\right]$ as in \eqref{eq:E[Z] upper}. 
\end{proof}

\subsection{Exact Analysis of JFFC}\label{appendix:Analysis of JFFC for K=2}

\begin{figure}[!t]
   \centerline{\includegraphics[width=.95\linewidth]{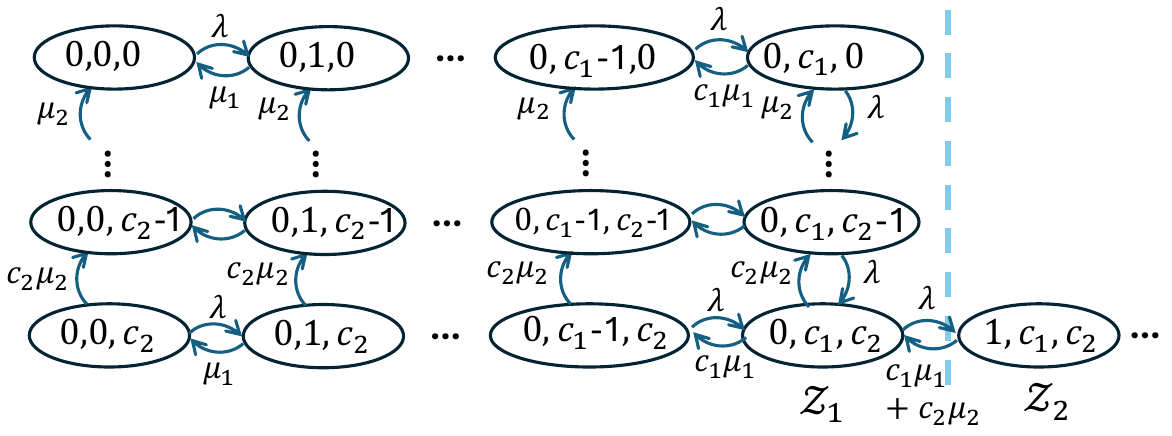}}
   \vspace{-1em}
    \caption{Transition diagram of the system state under JFFC for $K=2$.
    }
    \label{fig:CTMC_K2}
    \vspace{-.05em}
\end{figure}

Fig.~\ref{fig:CTMC_K2} illustrates the transition diagram of $(\bm{Z}(t))_{t\geq 0}$ in the case of $K=2$. Note that although JFFC always assigns each new arrival to the fastest available chain, it is still possible for the system to be in a state with $z_2>0$ and $z_1<c_1$ due to departures, which gives the transition diagram a 2D-grid structure as shown in Fig.~\ref{fig:CTMC_K2}. 

Define coefficient $\alpha_{\bm{z}}:=\pi_{\bm{z}}/\pi_{0,0,c_2}$ ($\forall \bm{z}\in \mathcal{Z}$), which is valid as Lemma~\ref{lem:ergodicity} implies $\pi_{0,0,c_2}>0$ for any $\lambda< \sum_{l=1}^K c_l \mu_l$. 
First, by definition we have $\alpha_{0,0,c_2}=1$, and by the flow balance equation for the states in $\{(0,i,c_2)\}_{i=0}^{n-1}$ we have
\begin{align}
\alpha_{0,n,c_2} = {1\over n\mu_1}\left(c_2\mu_2\sum_{i=0}^{n-1}\alpha_{0,i,c_2} + \lambda \alpha_{0,n-1,c_2} \right),~~~\forall n=1,\ldots,c_1.
\end{align}
Meanwhile, the state transition in $\mathcal{Z}_2$ has a birth-death structure with birth rate $\lambda$ and death rate $\nu :=\sum_{l=1}^K c_l \mu_l$, which implies that $\pi_{n,c_1,c_2} = \left({\lambda / \nu}\right)^n \pi_{0,c_1,c_2}$ and thus
\begin{align}\label{eq:pi_n,c1,c2}
\alpha_{n,c_1,c_2}  = \left({\lambda\over \nu}\right)^n \alpha_{0,c_1,c_2},~~~~~\forall n\in \mbbZ^+.  
\end{align}

Now suppose that $(\alpha_{0,i,z_2+1})_{i=0}^{c_1}$ has been computed (initially $z_2 = c_2-1$). By the flow balance between $\{\bm{z}':\: z'_2\leq z_2\}$ and $\{\bm{z}':\: z'_2 > z_2\}$, we have
\begin{align}\label{eq:alpha_0,c1,z2}
\alpha_{0,c_1,z_2} = {(z_2+1)\mu_2\over \lambda} \sum_{i=0}^{c_1}\alpha_{0,i,z_2+1}.
\end{align}
The flow balance equation for the states in $\{(0,i,z_2)\}_{i=0}^{n-1}$ is
\begin{align}
z_2\mu_2\sum_{i=0}^{n-1}\alpha_{0,i,z_2} + \lambda \alpha_{0,n-1,z_2} = (z_2+1)\mu_2\sum_{i=0}^{n-1}\alpha_{0,i,z_2+1} + n \mu_1\alpha_{0,n,z_2}, \label{eq:flow balance for z2}
\end{align}
which implies that $\alpha_{0,n,z_2}$ is an affine function of $\alpha_{0,0,z_2}$. Define parameters $\beta_{0,n,z_2}$ and $\gamma_{0,n,z_2}$ such that $\alpha_{0,n,z_2} = \beta_{0,n,z_2} \alpha_{0,0,z_2} + \gamma_{0,n,z_2}$; $\beta_{0,0,z_2}:=1$, $\gamma_{0,0,z_2}:=0$. For $n=1,\ldots,c_1$, \eqref{eq:flow balance for z2} implies the following recursive equations for $\beta_{0,n,z_2}$ and $\gamma_{0,n,z_2}$:
\begin{align}
\beta_{0,n,z_2} &= {1\over n\mu_1}\left(z_2\mu_2\sum_{i=0}^{n-1}\beta_{0,i,z_2} + \lambda\beta_{0,n-1,z_2}\right),\\
\gamma_{0,n,z_2}&={1\over n\mu_1}\left(z_2\mu_2\sum_{i=0}^{n-1}\gamma_{0,i,z_2} + \lambda \gamma_{0,n-1,z_2} - (z_2+1)\mu_2\sum_{i=0}^{n-1}\alpha_{0,i,z_2+1} \right).
\end{align}
In particular, since $\alpha_{0,c_1,z_2}$ is already known through \eqref{eq:alpha_0,c1,z2}, we can compute $\alpha_{0,0,z_2}$ by 
\begin{align}
\alpha_{0,0,z_2} = {1\over \beta_{0,c_1,z_2}}(\alpha_{0,c_1,z_2}-\gamma_{0,c_1,z_2}),
\end{align}
and hence $\alpha_{0,n,z_2}$ for $n=1,\ldots,c_1-1$. Repeating these steps for $z_2=c_2-1, \ldots,0$ yields $(\alpha_{0,i,z_2})_{i=0}^{c_1}$ for any $z_2<c_2$.

Since $\sum_{\bm{z}\in\mathcal{Z}}\pi_{\bm{z}} = 1$, we can compute the steady-state distribution by $\pi_{0,0,c_2} = \left(\sum_{\bm{z}\in\mathcal{Z}}\alpha_{\bm{z}}\right)^{-1}$ and $\pi_{\bm{z}} = \pi_{0,0,c_2}\alpha_{\bm{z}}$. This yields the steady-state mean system occupancy (recall $\nu :=\sum_{l=1}^K c_l \mu_l$)
\begin{align}
\mbbE_{\bm{\pi}}\left[\sum_{l=0}^K Z_l\right] &= {\sum_{\bm{z}\in\mathcal{Z}}\alpha_{\bm{z}}(\sum_{l=0}^K z_l) \over \sum_{\bm{z}\in \mathcal{Z}}\alpha_{\bm{z}}} \nonumber\\
&\hspace{-4em}= {\sum_{z_1=0}^{c_1}\sum_{z_2=0}^{c_2} \alpha_{0,z_1,z_2}(z_1+z_2) + {\lambda \alpha_{0,c_1,c_2}\over \nu-\lambda}({\nu\over \nu-\lambda}+c_1+c_2) \over \sum_{z_1=0}^{c_1}\sum_{z_2=0}^{c_2}\alpha_{0,z_1,z_2} + {\lambda \alpha_{0,c_1,c_2}\over \nu-\lambda}},
\end{align}
plugging which into \eqref{eq:steady-state mean response time} gives the steady-state mean response time. 

\end{document}